\documentclass[conference]{IEEEtran}
\IEEEoverridecommandlockouts
\usepackage{cite}
\usepackage{amsmath,amssymb,amsfonts}
\usepackage{graphicx}
\usepackage{textcomp}
\usepackage{xcolor}
\usepackage{graphicx}
\usepackage{xspace}
\usepackage{amsfonts}
\usepackage{color}
\usepackage{mathtools}
\usepackage{booktabs}
\usepackage{multirow}
\usepackage{balance}
\usepackage{url}
\usepackage{listings}
\usepackage{enumitem}
\usepackage{latexsym}
\usepackage{makecell}
\usepackage{float}
\usepackage{threeparttable}
\usepackage{arydshln}
\usepackage{pifont}
\usepackage{subfig}
\usepackage{algorithm}  
\usepackage{algorithmicx}  
\usepackage{algpseudocode} 
\usepackage[capitalise]{cleveref}

\usepackage{xcolor}
\usepackage{enumitem}

\usepackage{amsthm}
\newtheorem{theorem}{Theorem}

\usepackage{xcolor}
\usepackage[switch]{lineno}

\newcommand{\ie}{\emph{i.e.,}\xspace}
\newcommand{\eg}{\emph{e.g.,}\xspace}

\newcommand{\btitle}[1]{\vspace{1ex}\noindent \textbf{#1}}
\newcommand{\eat}[1]{}

\newcommand{\bcirclednumber}[1]{\ding{\numexpr181+#1\relax}}

\def\BibTeX{{\rm B\kern-.05em{\sc i\kern-.025em b}\kern-.08em
    T\kern-.1667em\lower.7ex\hbox{E}\kern-.125emX}}

\begin{document}

\title{Bala-Join: An Adaptive Hash Join for Balancing Communication and Computation in Geo-Distributed SQL Databases}

\author{\IEEEauthorblockN{Wenlong Song}
\IEEEauthorblockA{\textit{School of Cyber Engineering} \\
\textit{Xidian University}\\
Xi'an, China \\
songwl@stu.xidian.edu.cn}
\and
\IEEEauthorblockN{Hui Li}
\IEEEauthorblockA{\textit{School of Computer Science and Technology} \\
\textit{Xidian University}\\
Xi'an, China \\
hli@xidian.edu.cn}
\and
\IEEEauthorblockN{Bingying Zhai}
\IEEEauthorblockA{\textit{School of Computer Science and Technology} \\
\textit{Xidian University}\\
Xi'an, China \\
zhaibingying@stu.xidian.edu.cn}
\and
\IEEEauthorblockN{Jinxin Yang}
\IEEEauthorblockA{\textit{School of Cyber Engineering} \\
\textit{Xidian University}\\
Xi'an, China \\
yangjx@stu.xidian.edu.cn}
\and
\IEEEauthorblockN{Pinghui Wang}
\IEEEauthorblockA{\textit{School of Cyber Science and Engineering} \\
\textit{Xi’an Jiaotong University}\\
Xi'an, China \\
phwang@mail.xjtu.edu.cn}
\and
\IEEEauthorblockN{Luming Sun}
\IEEEauthorblockA{\textit{Yunxi Technology Company Ltd.}\\
Shanghai, China \\
sunluming@inspur.com}
\and
\IEEEauthorblockN{Ming Li}
\IEEEauthorblockA{\textit{Shandong Inspur Database Technology Company Ltd.}\\
Jinan, China \\
liming2017@inspur.com}
\and
\IEEEauthorblockN{Jiangtao Cui}
\IEEEauthorblockA{\textit{School of Computer Science and Technology} \\
\textit{Xidian University}\\
Xi'an, China \\
cuijt@xidian.edu.cn}
}

\maketitle

\begin{abstract}
Shared-nothing geo-distributed SQL databases, such as CockroachDB, are increasingly vital for enterprise applications requiring data resilience and locality. However, we encountered significant performance degradation at the customer side, especially when their deployments span multiple data centers over a Wide Area Network (WAN). Our investigation identifies the bottleneck in the performance of the Distributed Hash Join (\textsf{Dist-HJ}) algorithm, which is contingent upon a crucial balance between communication overhead and computational load. This balance is severely disrupted when processing skewed data from real-world customer workloads, leading to the observed performance decline.

To tackle this challenge, we introduce Bala-Join, an adaptive solution to balance the computation and network load in Dist-HJ execution. Our approach consists of the Balanced Partition and Partial Replication (BPPR) algorithm and a distributed online skewed join key detector. The former achieves balanced redistribution of skewed data through a multicast mechanism to improve computational performance and reduce network overhead. The latter provides real-time skewed join key information tailored to BPPR. Furthermore, an Active-Signaling and Asynchronous-Pulling (ASAP) mechanism is incorporated to enable efficient, real-time synchronization between the detector and the redistribution process with minimal overhead. Empirical study shows that Bala-Join outperforms the popular \textsf{Dist-HJ} solutions, increasing throughput by 25\%-61\%.
\end{abstract}


\section{Introduction}
\label{sec:intro}

Shared-nothing~\cite{stonebraker1986case} DBMSs such as CockroachDB \cite{taft2020cockroachdb}, TiDB \cite{huang2020tidb}, Spanner \cite{corbett2013spanner} and OceanBase\cite{yang2022oceanbase} are increasingly deployed to support geo-distributed applications that allocate data across various servers or data centers. In customer deployments built upon CockroachDB, our joint database research and development team with Shandong Inspur Database Technology Company Ltd. has encountered a significant challenge involving distributed query processing across wide-area networks (WANs). Specifically, for complex analytical queries involving distributed hash joins (Dist-HJ) that span data centers in different cities, the combination of real-world data skew and the high latency/limited bandwidth of the WAN creates a critical performance bottleneck.

This challenge is a direct consequence of the distributed hash join's execution model. As an essential method to implement the join operator, \textsf{Hash Join} relies on the hash table to efficiently match rows of two relations. The smaller relation is built into a hash table, while the bigger one is probed to match the tuples. At this point, the two relations are called the \textbf{build table}
and the \textbf{probe table}, respectively. \eat{In the distributed environment, both storage and computation differ significantly from the single-node environment. From the data \textit{storage} perspective, each table tends to be distributed across multiple data nodes; from a \textit{computational} aspect, the distributed environment comprises multiple compute nodes. Thus, the execution of distributed hash join needs to be redesigned. For convenience of description, we functionally abstract the nodes in the distributed environment into \textit{data}, \textit{compute}, and \textit{response} nodes. \textit{Data} nodes temporarily store the shards of the build and probe tables, which could be either original tables or intermediate results. \textit{Compute} nodes perform local \textsf{Hash Join} computation on the tuples acquired from the \textit{data} nodes. \textit{Response} nodes aggregate and return the result.} A general distributed \textsf{Hash Join} typically involves the following three steps:

\textbf{Redistribution}:
While the build and probe tables reside on various \textit{data} nodes, the tuples need to be redistributed to multiple \textit{compute} nodes for \textsf{Hash Join} tasks.

\textbf{Computation}:
Each \textit{compute} node performs local \textsf{Hash Join} computations on the received shares of tuples.

\textbf{Aggregation}:
\textit{Response} node aggregates the results of each \textit{ compute} node while there is no follow-up task, \eg another \textsf{Hash Join}. Otherwise, the \textit{compute} nodes subsequently transform into \textit{data} nodes, providing input for subsequent operations.

The default hash partitioning strategies used by the underlying CockroachDB foundation, while effective for local networks, often fail to adequately balance the trade-off between network communication load and overall query response time, especially in the presence of data skew. Consider a typical join query within the TPC-H benchmark shown below, the $customer$ table and the $orders$ table are joined on the customer key (\textit{CUSTKEY}). Orders from the same customer are distributed to the same \textit{compute} node. If some customers hold a substantial number of orders, \ie data skew, the corresponding compute nodes will undertake much more computational tasks than the other nodes. 

\begin{lstlisting}[language=SQL]
SELECT COUNT(*) 
FROM CUSTOMER JOIN ORDERS 
ON O_CUSTKEY = C_CUSTKEY;
\end{lstlisting}
\cref{fig:impact} provides a real example executing the above query in CockroachDB after introducing a 10\% data skew in the join between the $customer$ and $orders$ tables. The execution times of the three compute nodes (HashJoiners) are \textit{72ms, \textbf{9.1s}, and 69ms}, respectively. The data skew results in an unbalanced distribution of computational load across nodes, with one node bearing significantly higher computational overhead, thereby becoming the performance bottleneck of the overall query execution.

Extensive studies~\cite{dewitt1992practical,metwally2005efficient,xu2008handling,barthels2015rack,rodiger2016flow} have justified the aforementioned problem as a universal phenomenon in practice. Several efforts~\cite{xu2008handling,polychroniou2014track,rodiger2016flow, yang2023one} have been conducted to address this issue. The common idea is to employ different redistribution strategies for skewed and non-skewed data, respectively. This relies on the identification of skewed data, either through pre-collected statistics or online detection. Additionally, large-scale data transmission across nodes over a high-latency, limited-bandwidth WAN can lead to performance degradation\cite{pu2015low,pradhan2024optimal}, particularly when using the broadcast strategy commonly employed by existing solutions. Despite their promising progress in alleviating the load imbalance, the overall performance is still subject to the original data distribution across nodes~\cite{xu2008handling} or favorable network conditions~\cite{yang2023one}. For instance, there exist many queries joining tables spanned over Beijing and Shanghai (with a distance larger than 1000$km$) in Inspur group, a giant Chinese company providing high-performance computational infrastructure. We observe that in these cross-region joins, network overhead remains a critical factor for the join performance even if we replace the default \textsf{Dist-HJ} implementation with a series of alternatives tailored to load balance~\cite{xu2008handling,polychroniou2014track,rodiger2016flow, yang2023one}. The reason is that these solutions (\ie PRPD~\cite{xu2008handling} and PnR~\cite{yang2023one}) typically focus on either network overhead or balance, thus limiting their overall performance. 

\begin{figure}
\centering
\includegraphics[width=\linewidth]{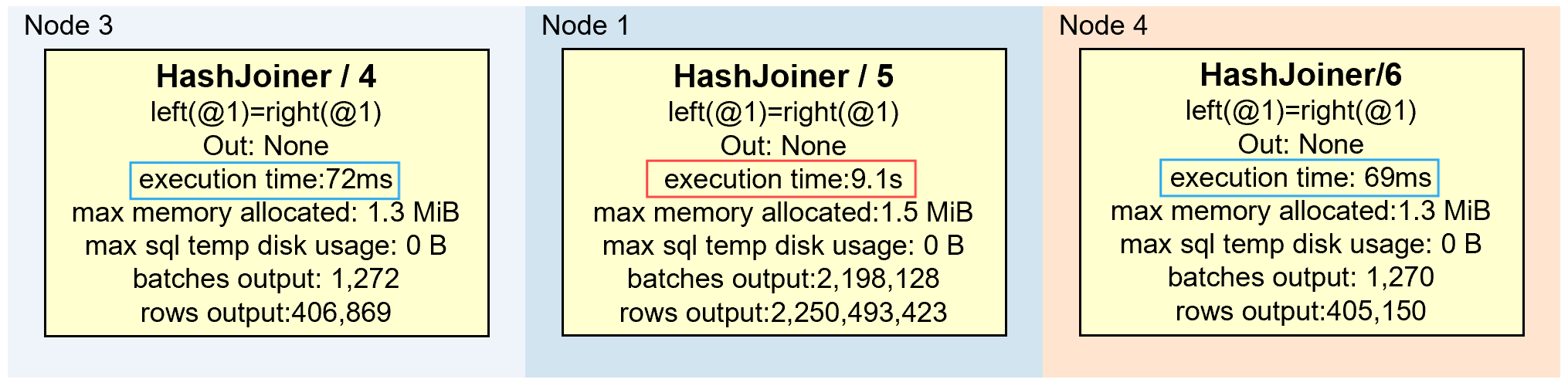}
\vspace{-1ex}\caption{Impact of data skew on distributed hash join}
\label{fig:impact}\vspace{-1ex}
\end{figure}

In this paper, we propose Bala-Join, which consists of a novel redistribution strategy and an online skewed join key detector. The former is adaptable to various skew scenarios and consistently performs well, while the latter can identify skewed join keys within intermediate join tables at runtime. In summary, our contributions are as follows.
\begin{itemize}
    \item We present a novel redistribution strategy (BPPR) for \textsf{Dist-HJ}, with theoretically guaranteed balance. By introducing a balance factor and a multicast mechanism, it ensures a minimized network overhead with guaranteed balance load (Section \ref{sec:bala-join}), 
    \item We present a distributed mechanism for skew detection, while minimizing additional overhead. The detection mechanism is further deeply integrated with BPPR. (Section \ref{sec:det}),
    \item Empirical study demonstrates the performance of our solution in balancing the load and network overhead. (Section \ref{sec:exp}).
\end{itemize}

\section{RELATED WORKS}
\label{sec:rel}

\subsection{Skewed data detection}
Skewed data detection strategies generally fall into two paradigms: static and dynamic. 
Static detection is typically conducted during the query planning phase. It relies on pre-computed statistics collected via analyzing histograms or sampling column values~\cite{dewitt1992practical,kitsuregawa1990bucket,polychroniou2014track,vitorovic2016load}. However, static approaches are rendered ineffective when handling intermediate results, where data distributions are unknown prior to execution.

Dynamic detection, in contrast, operates at runtime to identify frequent items from data streams with unpredictable distributions. This problem is formally equivalent to frequent item-set mining. Existing dynamic algorithms can be broadly categorized into: 1) \textbf{Sampling-based algorithms} \cite{gibbons1998new,manku2002approximate,demaine2002frequency}, which estimate frequency by inspecting a subset of the data stream; 2) \textbf{Counter-based algorithms} \cite{manku2002approximate,metwally2005efficient}, such as Space Saving, which utilize deterministic counters to track item occurrences; 3) \textbf{Sketch-based algorithms} \cite{charikar2004finding,cormode2005improved}, such as Count-Min Sketch\cite{cormode2005improved}, which employ probabilistic data structures for frequency estimation.

Several efforts have also been made in distributed dynamic detection. Solutions such as Parallel Space Saving \cite{cafaro2016parallel} and Double-Anonymous sketch \cite{zhao2023double} aggregate local estimates to form a global view. However, simplistic aggregation often necessitates multiple passes over the data or incurs high synchronization overhead, a challenge Bala-Join addresses via its single-pass ASAP mechanism.

\subsection{Joins over skewed data}
Data skew in distributed joins manifests in various forms, including the Tuple Placement Skew (TPS), the Selectivity Skew (SS), the Redistribution Skew (RS) and the Join Product Skew (JPS)\cite{walton1991taxonomy,dewitt1992practical}. While TPS is often mitigated by hash partitioning and SS depends heavily on query selectivity, RS and JPS remain the primary bottlenecks in parallel join processing\cite{xu2008handling,rodiger2016flow}. Existing mitigation strategies can be classified into three distinct approaches:

\textbf{Task scheduling approaches} formulate skew handling as a traditional task scheduling problem. In \textsf{Dist-HJ}, the join tables are divided into smaller partitions (\eg range partitions \cite{dewitt1992practical} or hash buckets \cite{kitsuregawa1990bucket,hua1991handling,alsabti2001skew,wolf1991effective,wolf1993parallel,wolf1994new,dewan1994predictive}) and assigned to processing units. Since task scheduling problems are NP-complete, various well-known heuristic algorithms (\eg LPT \cite{graham1969bounds} and MULTIFIT \cite{coffman1978application}) have been employed in such algorithms \cite{dewitt1992practical,kitsuregawa1990bucket,hua1991handling,alsabti2001skew,wolf1991effective,wolf1993parallel,wolf1994new}. However, these approaches typically rely on static cost estimates derived from tuple counts, often failing to account for dynamic runtime variances.

\textbf{Runtime adaptive scheduling approaches} dynamically adjust workload distribution during execution to handle load imbalances~\cite{dewitt1992practical,alsabti2001skew,dewan1994predictive,shatdal1993using,kitsuregawa1995dynamic,zhou1995handling,zhou2019fastjoin}. 
Instead of predicting execution time in advance, these methods adapt in real-time by shifting tasks from overloaded nodes to underutilized ones.
Common scheduling strategies include Virtual Processor Scheduling (VPS)~\cite{dewitt1992practical}, work-stealing via shared virtual memory~\cite{shatdal1993using}, and multi-stage scheduling managed by a central processor~\cite{zhou1995handling}.
To guide these decisions, approaches rely on various cost metrics, such as estimated join costs~\cite{dewitt1992practical}, output and work functions~\cite{alsabti2001skew}, or real-time processing speeds~\cite{kitsuregawa1995dynamic}.
Additionally, some solutions are tailored for heterogeneous configurations~\cite{dewan1994predictive} or stream processing environments, such as FastJoin~\cite{zhou2019fastjoin}.
Despite their adaptability, these methods often incur significant state migration overhead.

\textbf{Skew-aware partitioning approaches}, most relevant to our work, optimize physical data placement to mitigate computational imbalance. 
PRPD~\cite{xu2008handling} employs a strategy where skewed tuples remain local while the corresponding partition of the opposing table is broadcast, relying on the assumption of random data distribution. 
Flow-Join~\cite{rodiger2016flow} extends this by applying the Symmetric Fragment Replicate (SFR) strategy to handle bilateral skew, demonstrating efficacy in high-bandwidth environments. 
Other contributions focus on specific optimization goals: Track Join~\cite{polychroniou2014track} minimizes traffic by optimizing transmission plans for unique values, while FGSD~\cite{gavagsaz2019load} provides fine-grained load balancing specifically for MapReduce workflows. 
More recently, PnR~\cite{yang2023one} shifts the paradigm from minimizing network transfer to achieving near-perfect computational balance, prioritizing CPU efficiency in high-speed networks. 
While these strategies have made significant strides, they often trade network efficiency for load balance (\eg PnR) or rely on the assumption of uniform skew distribution (\eg PRPD). 
In contrast, Bala-Join is designed to decouple load balancing from the original data distribution while minimizing the network overhead associated with broadcasting.

\section{PRELIMINARY and MOTIVATION}
\label{sec:pre}

\subsection{Shared-nothing DBMSs}
The shared-nothing architecture is widely adopted for distributed DBMSs, where each node operates independently with standalone CPU, memory, and storage.
In these DBMSs, \textsf{Dist-HJ} is executed in parallel. Since data shards reside on different nodes, tuples with the same join key need to be redistributed to the same compute node. Each compute node then performs a local hash join, after which the results are aggregated by the response node. Although hash partitioning is the most commonly used redistribution strategy in practice, it can lead to significant imbalance under data skew. In such cases, newly added nodes cannot share the skewed data from the hot nodes, thus limiting the scalability.
For instance, as a representative shared-nothing DBMS, CockroachDB~\cite{taft2020cockroachdb} organizes data into ranges that are distributed and replicated throughout the cluster. When executing \textsf{Dist-HJ}, CockroachDB dynamically selects range replicas to generate the physical execution plan. Thus, the participating nodes for a given join query are not predetermined, and the volume of data contributed by each node may vary substantially. Moreover, the default \textsf{Dist-HJ} in CockroachDB employs hash partitioning as the exclusive redistribution strategy. Therefore, when data skew occurs, it can cause a load imbalance across the cluster.


\subsection{Distributed Hash Join}

\begin{figure*}[htbp]
\centering
\subfloat[PRPD algorithm]{
    \includegraphics[width=0.26\linewidth]{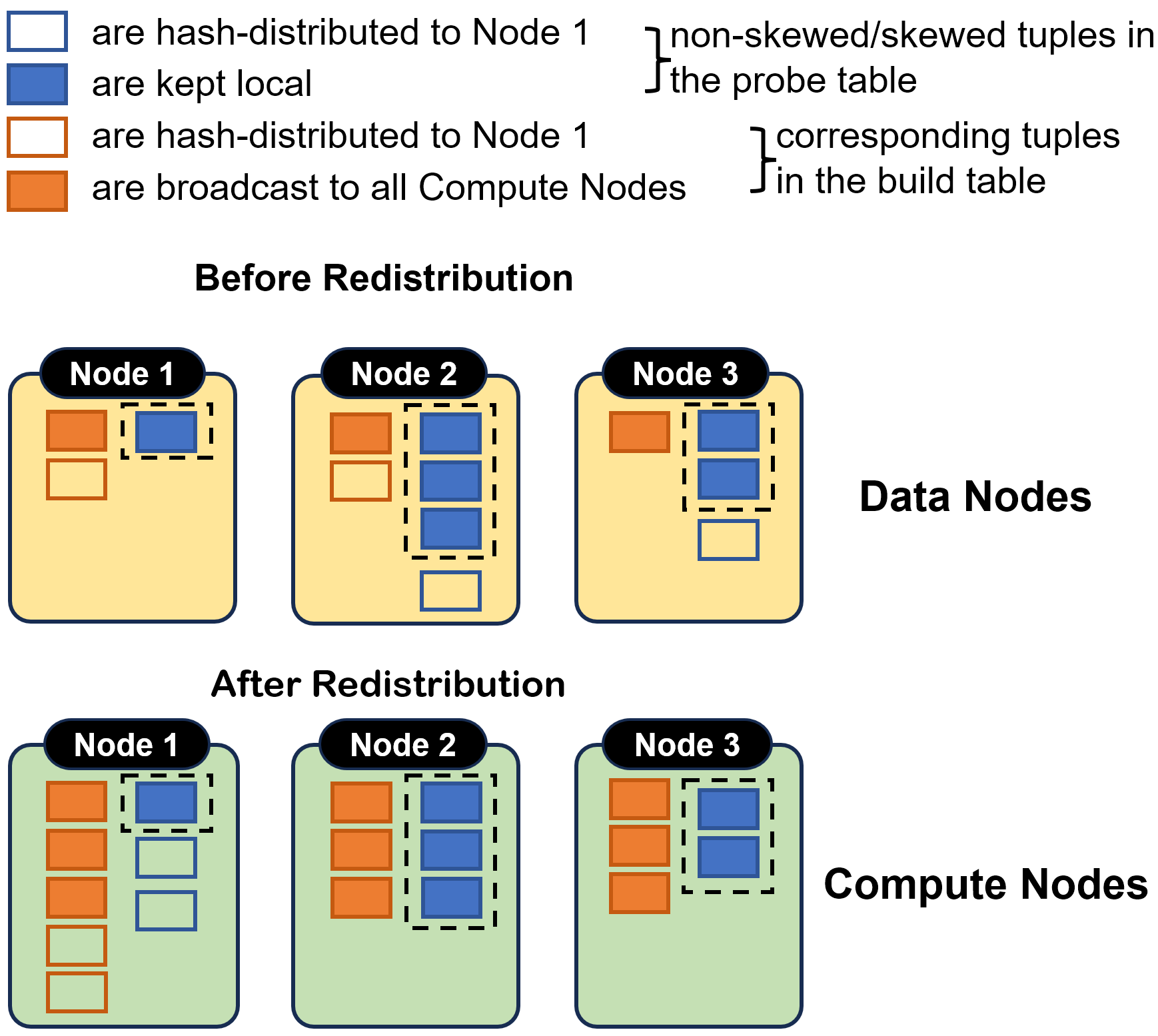}
    \label{fig:PRPD}
}
\subfloat[SFR algorithm]{
    \includegraphics[width=0.42\linewidth]{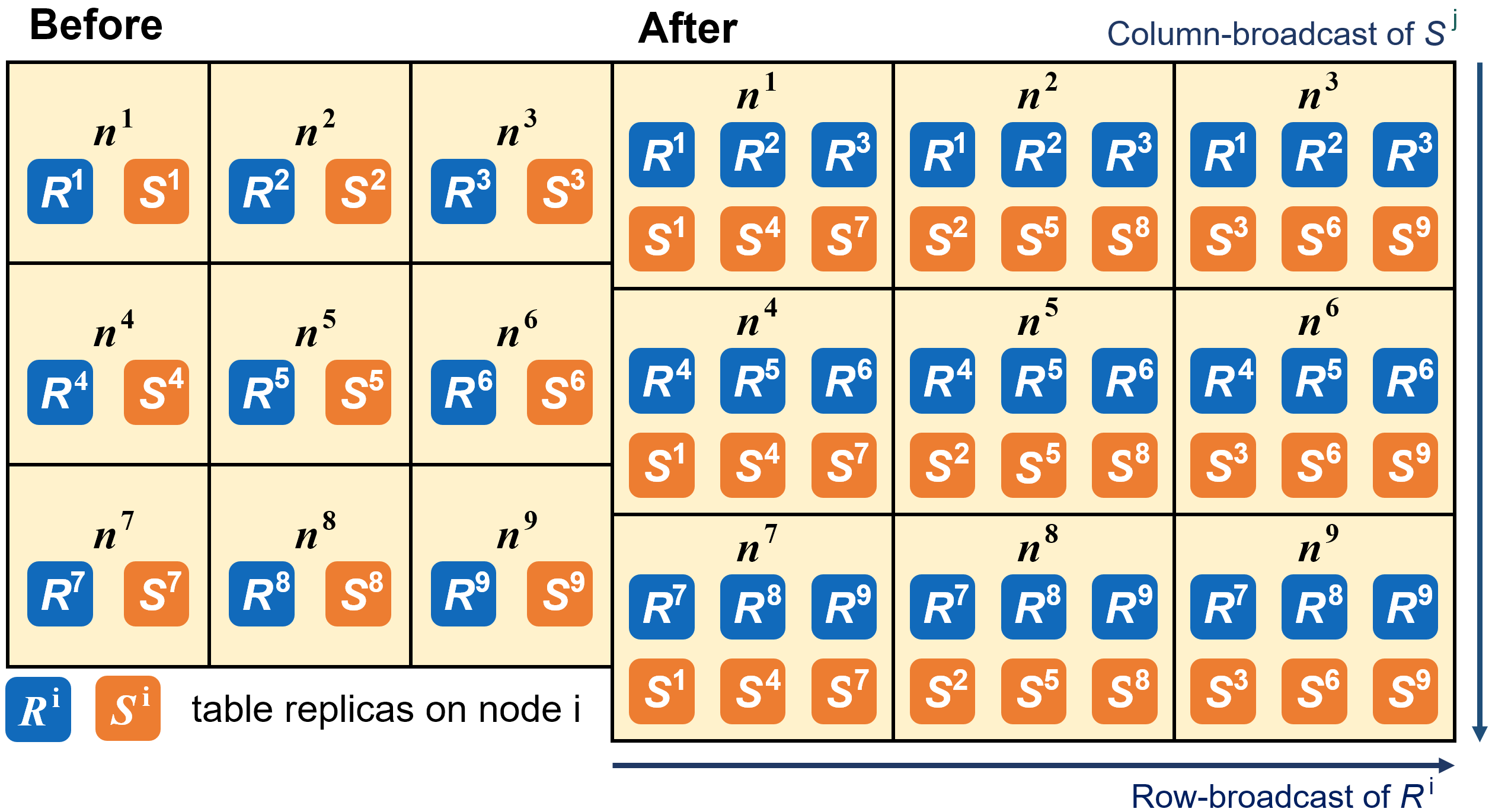}
    \label{fig:SFR}
}
\subfloat[PnR algorithm]{
    \includegraphics[width=0.265\linewidth]{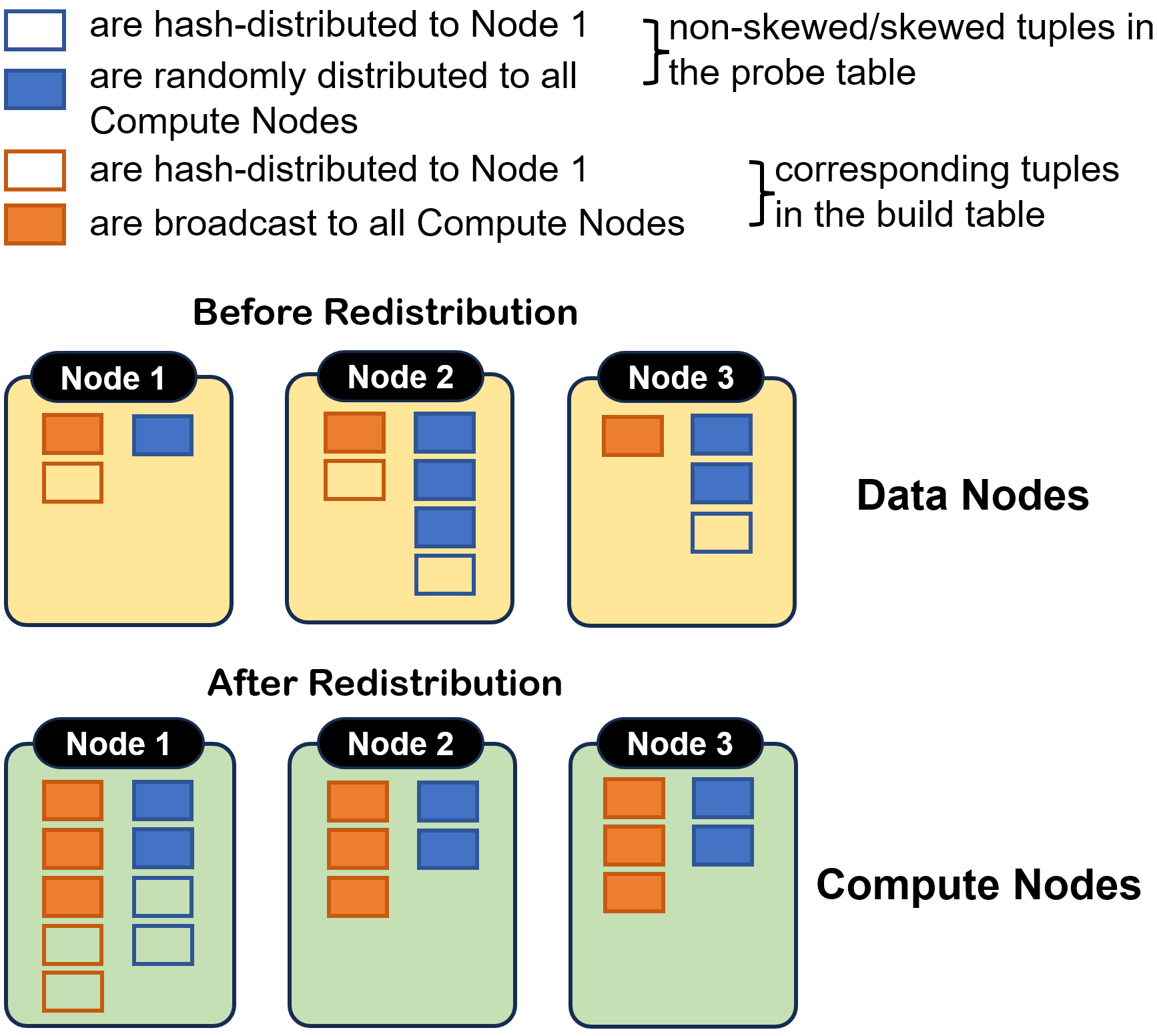}
    \label{fig:PnR}
}
\caption{Redistribution strategies overview}
\label{fig:overview}\vspace{-1ex}
\end{figure*}
This part explores the implementation of distributed hash join, focusing on several essential algorithms that optimize performance under different data distribution scenarios.

\textbf{Grace Hash Join}: As the standard implementation for shared-nothing architectures, \textsf{GraHJ}~\cite{kitsuregawa1983application} partitions both the build table $R$ and probe table $S$ across nodes using a consistent hash function. This redistribution ensures that tuples with identical join keys are co-located on the same processing unit. Consequently, the global join operation is decomposed into a set of independent, parallel local hash joins, the union of which constitutes the final result.

\textbf{PRPD}: \textsf{GraHJ} can experience degradation of efficiency under data skew. PRPD deals with skewed tuples and non-skewed tuples individually: the former are kept local while the latter follow a hash partitioning strategy. The corresponding tuples with skewed join keys in the build table would be broadcast to all nodes in the cluster. PRPD is designed to minimize network overhead and achieve a balanced load. Its optimal performance is based on the assumption that the skewed tuples are distributed uniformly across the nodes. However, due to the unpredictable distribution of skewed tuples, it is difficult to consistently achieve the expected favorable performance.

\eat
{
\cref{PRPD} illustrates the core idea of the PRPD algorithm. In the given example, $node\ 0$ and $node\ 1$ act as both \textit{data} nodes and \textit{compute} nodes. $S$ contains 10 tuples, among which the tuples of join key `1' are skewed. The skewed tuples in $S$ in both $node\ 0$ and $node\ 1$ are retained locally, while the skewed tuples of the join key `1' of table $R$ are replicated in both $node\ 0$ and $node\ 1$. The rest tuples in $R$ and $S$ are redistributed according to \textsf{GraHJ}.
}

\textbf{SFR}: SFR is the redistribution strategy adopted by Flow-Join~\cite{rodiger2016flow}. For tuples that are skewed only in the build table or the probe table, Flow-Join adopts the PRPD strategy. For tuples that are skewed in both tables, Flow-Join employs the Symmetric Fragment Replicate (SFR) strategy, where the nodes in the cluster are logically arranged in a matrix. The tuples from the build table are replicated across the matrix horizontally, while the corresponding tuples from the probe table are replicated to nodes vertically.

\textbf{PnR}: PnR~\cite{yang2023one} behaves similarly to Flow-Join, but selects to distribute the skewed tuples randomly or in a round-robin manner across all nodes in the cluster. The corresponding tuples from the probe table are broadcast across the cluster.

\cref{fig:overview} shows an overview of the redistribution strategies designed for data skew. In the following, we shall discuss the applicable scenarios and limitations of these approaches, thereby clarifying our objectives.

\subsection{Observations and Targets}
\label{subsec:observations}
PRPD and SFR focus on optimizing network overhead, while PnR emphasizes load balance. Both groups have their advantages in specific scenarios, but may suffer from limitations in more generalized situations.

As PRPD keeps the skewed tuples local, some nodes still bear a higher load than others when the original distribution of the skewed tuples is not uniform across nodes. Similarly, in SFR, when a row or column in the matrix contains many more skewed tuples than others, the nodes on that row or column will be overloaded. That is, PRPD and SFR are not suitable for scenarios where the original distribution of skewed tuples is excessively imbalanced.

\eat{
\cref{PRPD-limit} shows how PRPD performs under different data distributions. Here, the vertical bars indicate the count of specific join keys in the probe and build tables, which showcase three types of data distribution scenarios. Clearly, when the difference in volume between the build table and the probe is not too large (\eg of identical magnitude), it results in substantial network overhead.
}

PnR achieves balanced redistribution across various scenarios at the cost of massive network overhead. Nodes with more skewed tuples bear a heavier load during the redistribution phase. Furthermore, PnR suffers from the same problem as PRPD, especially when the volumes of both tables are nearly identical.

PRPD, Flow-Join and PnR all suffer from limitations in either the computation or the redistribution phase. In summary, previous strategies face the following challenges:

\textbf{Chal. \bcirclednumber{1}} \textit{Dependence on the original distribution}: Any strategy keeping skewed tuples local would experience performance degradation when the original distribution of skewed tuples is non-uniform across nodes.

\textbf{Chal. \bcirclednumber{2}} \textit{Statistical dependencies}: Previous approaches require the identification of skewed tuples. Moreover, it is also essential for Flow-Join and PnR to determine the number of skewed tuples in both the build and the probe table. Such statistics are unavailable in streaming scenarios or intermediate results, rendering these methods inapplicable.

\textbf{Chal. \bcirclednumber{3}} \textit{Non-adaptive strategies}: Given a tuple, once the necessary statistical information is acquired, the redistribution will follow a specific strategy. Therefore, a dynamic load balance within the cluster is not achievable. In streaming scenarios, tuples determined to be non-skewed at the beginning may eventually become skewed as the data stream progresses. Therefore, the correctness of the distributed join is challenged.

In order to address the challenges above, we present Bala-Join, which dynamically detects and deals with skewed tuples. Bala-Join immediately distributes tuples upon arrival, without waiting for the complete statistical information. Bala-Join introduces a novel skew-aware distribution strategy, namely BPPR, which ensures dynamic join load balance in the cluster that is independent of the original data distribution, mitigating \textbf{Chal. \bcirclednumber{1}}. Tuples with the same join key in both tables, whether determined as skewed or not, will always be distributed to the same subset of nodes, which resolves the consistency issue in dynamic adaptation raised in \textbf{Chal. \bcirclednumber{3}}. Finally, the integration of BPPR with the distributed detector and the ASAP mechanism achieves real-time detection and distribution, significantly improving the performance of the distributed hash join and addressing \textbf{Chal. \bcirclednumber{2}}.

\section{BALA-JOIN}
\label{sec:bala-join}
{

\begin{table}[t]
  \caption{Notations}
  \label{tab:notations}
  \vspace{0ex}
  \footnotesize
  \renewcommand{\arraystretch}{1.25} 
  \begin{tabular}{c p{6cm}} 
    \toprule
    \textbf{Symbol} & \textbf{Description}\\
    \midrule
    $R, S$ & Build table and probe table\\
    $R^i, S^i$ & Local data shards of $R$ and $S$ initially located at node $i$ (data streams)\\
    $S_{skew}, S_{non}$ & Sets of skewed and non-skewed tuples in $S$\\
    $R_{skew}, R_{non}$ & Sets of build tuples whose keys match those in $S_{skew}$ and $S_{non}$\\
    $S_{skew}^i, S_{non}^i$ & Local subsets of skewed/non-skewed probe tuples at node $i$\\
    $R_{skew}^i, R_{non}^i$ & Local subsets of build tuples at node $i$\\
    $U(x)$ & Set of candidate destination compute nodes for skewed key $x$\\
    $C(t)$ & Selection function that routes a skewed tuple $t$ to a node in $U(t.key)$\\
    $\widehat{S}_{skew}^{j}$ & Total volume of skewed probe tuples received by compute node $j$\\
    $\widehat{S}_{skew}^{ij}$ & Total volume of skewed tuples from source node $i$ distributed to target node $j$\\
    $B$ & Balance factor indicating load disparity (Eq.~\ref{equation:B})\\
    $\epsilon$ & User-defined threshold for load imbalance ($B \le \epsilon$)\\
    \bottomrule
  \end{tabular}
\end{table}

\subsection{System Overview}
\label{sub:overview}

\begin{figure}
\centering
\includegraphics[width=\linewidth]{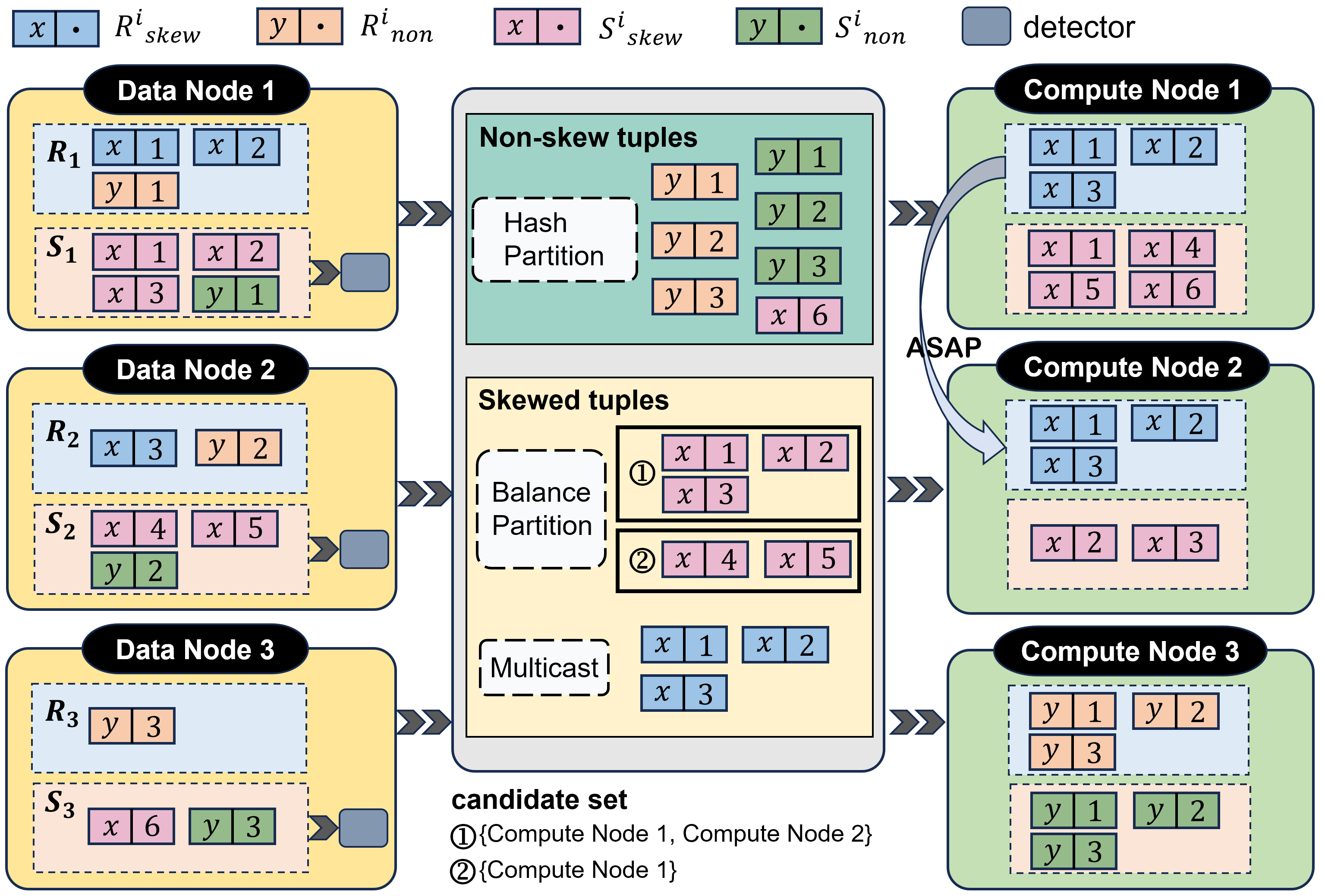}
\vspace{0ex}\caption{Overview of the Bala-Join}
\label{fig:Bala-Join}\vspace{0ex}
\end{figure}
Bala-Join targets distributed hash joins in geo-distributed databases, and focuses on two major costs: cross-region communication overhead and computation cost caused by load imbalance.

Bala-Join consists of a novel redistribution strategy, BPPR, and an online skewed join key detector. The build and probe tables participate in the join as data shards or data streams across data nodes. Upon the arrival of probe tuples, the distributed detectors deployed on the nodes check each tuple, determine whether it is skewed, and then redistribute it following BPPR.

To ensure the tuples are correctly joined, when a tuple from one table is distributed to a certain node, all tuples with the same join key in another table must also be distributed to that node. Since the probe table is typically much larger than the build table, balancing the skewed tuples from the build table would necessitate multicasting or broadcasting the massive probe table tuples. Furthermore, we focus on scenarios involving intermediate results or data streams, where statistics are unavailable to \bcirclednumber{1} compare the volume of skewed build tuples against their corresponding probe tuples and \bcirclednumber{2} decide how to deal with them (a strategy used by PnR). Consequently, we exclusively investigate the probe side.

In BPPR, each data node maintains an expandable target node set for each join key of the skewed tuples. Newly arrived skewed tuples are distributed to proper nodes in the sets so as to satisfy a predefined balance factor, which shall be described in detail in~\cref{sec:Balanced Partition Algorithm}. When none of the existing nodes are adequate, the node set is expanded. The tuples with the same join key in the build table will soon be multicast to the node set afterwards.

\noindent\underline{\textit{Example:}} \cref{fig:Bala-Join} shows an overview of Bala-Join. Data nodes continuously receive streams of build and probe tuples, where detectors on each node dynamically identify skewed probe tuples in real-time. We use subscripts $p$ and $b$ to denote tuples from the probe and build tables, respectively. Non-skewed tuples (\eg ${\langle x, 6 \rangle}_{p}$, ${\langle y, 1 \rangle}_{p} \sim {\langle y, 3 \rangle}_{p}$) are hash-partitioned to their corresponding compute nodes (\eg Compute Nodes 1 and 3), a strategy also applied by default to the build table. 
In contrast, skewed probe tuples are distributed to specific target node sets based on their join keys; for instance, tuples with key $x$ from Data Node 1 (\ie ${\langle x, 1 \rangle}_{p} \sim {\langle x, 3 \rangle}_{p}$) are distributed to \{Compute Node 1, Compute Node 2\}, while those from Data Node 2 (\ie ${\langle x, 4 \rangle}_{p}, {\langle x, 5 \rangle}_{p}$) are sent to \{Compute Node 1\}. 
To ensure correctness, the corresponding build tuples (\ie ${\langle x, 1 \rangle}_{b} \sim {\langle x, 3 \rangle}_{b}$) are multicast to all these destination nodes (\ie the union set \{Compute Node 1, Compute Node 2\}). 
Since build tuples are inherently unaware of the probe tuples' destinations, this matching is facilitated by the ASAP mechanism described in \cref{sec:det}. 
Furthermore, even if the same join key (\eg $x$) is inconsistently identified as skewed or non-skewed across different nodes, the target set update algorithm guarantees that all associated tuples are invariably distributed to the same consistent set of nodes.

\subsection{Notations}
\label{sec:notation}
Generally, we consider a distributed hash join operation $R \bowtie S$ on a $n$-node shared-nothing cluster, where the build table $R$ and probe table $S$ are initially partitioned into shards $R^i$ and $S^i$ at node $i$ ($0 \le i < n$). This abstraction treats inputs as distributed data streams, regardless of whether they are original tables or intermediate results.

We identify a tuple in the probe table $S$ as \textit{skewed} if its join key frequency exceeds a predefined threshold $\theta$. 
Based on this, $S$ is partitioned into the set of skewed tuples, denoted as $S_{skew}$, and the set of non-skewed tuples, $S_{non}$.
Correspondingly, the build table $R$ is logically divided into $R_{skew}$ and $R_{non}$, where $R_{skew}$ comprises tuples whose join keys appear in $S_{skew}$ (regardless of their intrinsic frequency in $R$).
On each node $i$, the local shards are further partitioned into four distinct subsets: $S_{skew}^i$, $S_{non}^i$, $R_{skew}^i$, and $R_{non}^i$.

\subsection{Balanced Partition and Partial Replication}
\label{subsec:bppr}
\vspace{.5ex}\noindent\textbf{BPPR Workflow.}

BPPR adheres to the general processing workflow of \textsf{Dist-HJ}, but introduces an additional tuple-level partitioning phase designed to further subdivide skewed tuples into more fine-grained partitions, to improve load balance. Assisted by the real-time skew detector (detailed in \cref{sec:det}), the input streams are dynamically partitioned and processed. For a join operation of the form ${R} \mathop{\bowtie}\limits_{R.a = S.b} {S}$, the general procedure is executed as follows.

\btitle{Step 1, Redistribution.} Given the above, $ R^i $ and $ S^i $ are divided into four parts: $ S_{non}^i $, $ S_{skew}^i $, $ R_{non}^i $, and $ R_{skew}^i $, respectively. These parts are further partitioned and redistributed as follows:

As defined in ~\cref{sec:notation}, the local shards $S^i$ and $R^i$ are logically divided into four subsets, which are distinguished by different colors in the Data Nodes of \cref{fig:Bala-Join}. These subsets are further partitioned and redistributed as follows:
\begin{itemize}
    \item $S_{non}^i$: Tuples in this subset are hash-partitioned and distributed, achieving nearly even distribution across $ n $ nodes.
    
     \underline{\textit{Example:}} ${\langle x, 6 \rangle}_{p}$, ${\langle y, 1 \rangle}_{p} \sim {\langle y, 3 \rangle}_{p}$ are distributed to Compute Node 1 and Compute Node 2, respectively. Note that although $x$ is a skewed key in $S$, the detector at Data Node 3 identifies it as non-skewed since it only monitors the local stream $S_3$. The BPPR and ASAP mechanisms shall guarantee the correctness of the join result even in such scenarios.

    \item $S_{skew}^i$: Each tuple $t\in S_{skew}^i $ is routed to a target node, determined by a selection function $C(t)$ within a candidate set $U(x)$. To quantify the degree of balance, we define the \textit{balance factor} $B$:
    \begin{equation}
    \label{equation:B}
	   B = ({\mathop{max}\limits_{j} |{\widehat{S}_{skew}}^{j}|} - {\mathop{min}\limits_{j} |{\widehat{S}_{skew}}^{j}|})/{\mathop{max}\limits_{j} |{\widehat{S}_{skew}^{j}}|}
    \end{equation}

    where ${\widehat{S}_{skew}}^{j}$ denotes the volume of skewed tuples received by \textit{compute} node $j$.
    
    \underline{\textit{Example:}} As \cref{fig:Bala-Join} shows, Compute Node 1 receives the maximum number of skewed tuples (4), while Compute Node 3 receives zero. Consequently, the global balance factor is calculated as $B = (4-0)/4 = 1$. Note that even though only $x$ is shown as a skewed key in the diagram, $\widehat{S}_{skew}^{j}$ accounts for \textit{all} skewed keys, rather than a specific key such as $x$.
    
    Naturally, a balanced redistribution of skewed tuples should satisfy $B\le\epsilon$, where $\epsilon$ is a predefined threshold. For each skewed tuple at arrival, $C(t)$ attempts to select a target node from the set of nodes and test whether it is balanced enough. If none of the nodes in $U(x)$ meets the requirement, $U(x)$ is expanded to include additional nodes for selection (shall be discussed later in \cref{algorithm:BalancedPartition}).
    
    \underline{\textit{Example:}} ${\langle x, 1 \rangle}_{p} \sim {\langle x, 3 \rangle}_{p}$ and ${\langle y, 4 \rangle}_{p} \sim {\langle y, 5 \rangle}_{p}$ are distributed to the set $U(x)=$ \{Compute Node 1, Compute Node 2\}. Note that $U(x)$ is effectively the union of $U^1(x)$ and $U^2(x)$ maintained by Data Nodes 1 and 2, a mechanism that will be detailed in \cref{sec:Balanced Partition Algorithm}.

    \item $ R^{i}_{skew} $: To ensure correctness, each tuple in $ R_{skew} $ with the join key $\boldsymbol{x}$ must be multicast to the nodes specified by $ U(x) $, which we refer to as \textit{Partial Replication}.

    \underline{\textit{Example:}} ${\langle x, 1 \rangle}_{b} \sim {\langle x, 3 \rangle}_{b}$ are multicast to \{Compute Node 1, Compute Node 2\} since ${\langle x, 1 \rangle}_{p} \sim {\langle x, 3 \rangle}_{p}$ are distributed to that set.

    \item $R_{non}^i$: Similar to $S_{non}^i$, these tuples are distributed via standard hash partitioning.

    \underline{\textit{Example:}} ${\langle y, 1 \rangle}_{b} \sim {\langle y, 3 \rangle}_{b}$ are distributed to Compute Node 3 to match ${\langle y, 1 \rangle}_{p} \sim {\langle y, 3 \rangle}_{p}$.
\end{itemize}

\eat{During the computation phase, each compute node receives the redistributed tuples from all data nodes according to \textbf{Step 1}. Consequently, for a compute node $i$, its inputs can be viewed as the union of four sets. For the non-skewed portion, it receives the hash-distributed probe tuples $\mathcal{S}^i_{hash}$ from $S_{non}$ satisfying $(hash(t.b)\bmod n)=i$, and the hash-distributed build tuples $\mathcal{R}^i_{hash}$ from $R_{non}$ satisfying $(hash(t.a)\bmod n)=i$. For the skewed portion, it receives the balanced-distributed skewed probe tuples $\mathcal{S}^i_{bal}$ from $S_{skew}$ with targets satisfying $C(t)=i$, and the multicast build tuples $\mathcal{R}^i_{bal}$ from $R_{skew}$ satisfying $i\in U(t.a)$. After receiving these inputs, node $i$ performs two local \textsf{Hash Join}s, namely $\mathcal{R}^i_{bal}\bowtie_{a=b} \mathcal{S}^i_{bal}$ and $\mathcal{R}^i_{hash}\bowtie_{a=b} \mathcal{S}^i_{hash}$, and then unions the two outputs as its local result.}

\btitle{Step 2, Computation.} Each compute node processes the tuples received from the redistribution step (\eg Compute Node 2 processes the skewed probe tuple ${\langle x, 2 \rangle}_{p} $ and $ {\langle x, 3 \rangle}_{p}$ against the multicast build tuple ${\langle x, 1 \rangle}_{b} \sim {\langle x, 3 \rangle}_{b}$). Node $j$ performs two independent local hash joins on the skewed and non-skewed streams: $R_{non}^j \bowtie S_{non}^j$ and $R_{skew}^j \bowtie S_{skew}^j$.

\btitle{Step 3, Union.} Upon completion of the computation step, each compute node produces a set of local partial results. 
Formally, the global outcome is derived by aggregating the outputs from all $n$ compute nodes. 
Since the probe tuples are disjointly partitioned (either via hashing or the $C(t)$ function), the global join result is simply the union of these local results:
\begin{equation}
    {R} \mathop{\bowtie} {S} = \bigcup^{n-1}_{j=0} \, \left( (R_{non}^j \bowtie S_{non}^j) \cup (R_{skew}^j \bowtie S_{skew}^j) \right)
\end{equation}
where $j$ represents the index of a compute node.

In contrast to strategies that enforce uniform distribution for \textit{every} skewed key across \textit{all} nodes (\eg PRPD and PnR), BPPR optimizes load balancing at the \textit{cluster level}. 
For instance, PnR scatters tuples with key $x$ to all $n$ nodes to minimize local skew, incurring high multicast overhead. 
BPPR, conversely, restricts the distribution of $x$ to a minimal subset $U(x)$. 
Consequently, although the load for a specific key $x$ may not be perfectly even across the entire cluster, the aggregate load of all skewed tuples is balanced (\ie satisfying $B \le \epsilon$). 
This design trade-off significantly reduces network cost while maintaining computational balance, a claim substantiated by our empirical evaluation in \cref{sec:exp}.

\vspace{.5ex}\noindent\textbf{Balanced Partition Algorithm}

\label{sec:Balanced Partition Algorithm}

Intuitively, according to the above strategy, for skewed tuples with a join key $\boldsymbol{x}$, each data node should send them to the same set of nodes, thus reducing the network overhead caused by multicasting the build table. 
However, in a distributed environment, achieving such consensus requires significant network communication. 
Specifically, for each arriving tuple, the system must synchronize the global balance state to decide whether to route it to the current $U(x)$ or broadcast a set expansion. This heavy synchronization overhead is prohibitive when handling intermediate results or data streams. To address this, we propose the \textit{Balanced Partition} algorithm based on a sequence generator, which achieves coordination-free consensus while ensuring load balance in the cluster.

In the absence of a sequence generator, each data node maintains a \textit{divergent} local set $U^i(x)$ for each skewed join key $\boldsymbol{x}$. 
This imposes two critical requirements:

(1) \textbf{Correctness}: The effective global target set must be the union of all local sets, \ie $U(x) = \bigcup^{n-1}_{i=0} U^i(x)$, to ensure that build tuples from $R_{skew}$ reach all potential destinations;

(2) \textbf{Efficiency}: The cardinality $|U(x)|$ must be minimized to reduce the multicast overhead.

Requirement (1) is guaranteed by the ASAP mechanism (detailed in \cref{ssec:asap}). 
Requirement (2) motivates the design of our deterministic sequence generator, described below. To ensure (2), the sets $U^i(x)$ for $i=0,\ldots , n-1$ should be constructed and expanded in a consistent manner, ensuring that the maximal set subsumes all others. This would minimize network overhead during multicast.
To fulfill that, we present the Node Set Update Algorithm, \ie \cref{algorithm:UpdateU}, which incorporates a deterministic sequence generator. For a skewed join key $\boldsymbol{x}$, the procedure \textsc{GenSeq} generates a consistent sequence of candidate nodes across all data nodes without requiring communication. The generation relies on a hash-based iterative mechanism. Specifically, it calculates a candidate node index using \( hash(\boldsymbol{x} + epoch) \bmod \boldsymbol{n} \) and iteratively increments the epoch to find non-duplicate nodes. Notably, the first element generated (where $epoch=0$) is always \( hash(\boldsymbol{x}) \bmod \boldsymbol{n} \), which corresponds to the target node for the tuple under the default hash distribution. This means that when \( \boldsymbol{x} \) is first identified as non-skewed and later identified as skewed, its target node will always remain in \( U(x) \). This resolves the issue raised in \textbf{Chal.} \bcirclednumber{3}.

\textsc{GenSeq} ensures consistency of sequences generated for the same skewed join key across various nodes. Based on this, each data node can independently expand \( U^i(x) \) following the same sequence. The procedure \textsc{UpdateU} in \cref{algorithm:UpdateU} demonstrates this update process. Whenever the node set requires expansion (and as long as \( |U^i(x)|<n \)), it invokes \textsc{GenSeq} to generate the next deterministic node in the sequence and adds it to \( U^i(x) \). This ensures that the multicast target set is expanded consistently across the cluster.

Given that each data node has achieved a consensus on the expansion of \( U^i(x) \), we move on to pursue load balance within the cluster. Recall in \cref{equation:B}, $S_{skew}^i$ shall be partitioned in a balanced manner such that the aggregate data volume transmitted to each compute node is balanced. We denote \( \widehat{S}_{skew}^{ij} \) as the partition of skewed tuples sent from data node \( i \) to compute node \( j \).
To ensure the partitions are balanced after partitioning, we define the \textit{local balance factor} $ B^i$: 
\begin{equation}
	B^i = ({\mathop{max}\limits_{j} |\widehat{S}_{skew}^{ij}|} - {\mathop{min}\limits_{j} |\widehat{S}_{skew}^{ij}|})/{\mathop{max}\limits_{j} |\widehat{S}_{skew}^{ij}|}
	\label{equation:BalanceFactor}
\end{equation}

The balance factor \( B^i \) measures the degree of balance in the \textit{aggregate} data volume across active target nodes. 
A smaller $B^i$ indicates that node $i$ distributes its skewed workload uniformly across the target compute nodes. 
Note that this metric accounts for the combined load of all skewed keys.
For instance, if skewed tuples with keys $k_1$ and $k_2$ from node $i$ are distributed to $\{\text{Compute Node 1}, \text{Compute Node 2}\}$ and $\{\text{Compute Node 3}\}$ respectively, with comparable volumes on each node, then $\max = \min$ and consequently $B^i = 0$. 
In this case, the system is balanced, and no expansion of candidate sets is required.
We set a threshold \( \epsilon \) such that \( B^i \leq \epsilon \).
\begin{theorem} 
\label{balafactor}
In the case of balanced partitioning, if the balance factor \( B^i\le\epsilon \), it follows that the overall balance factor \( B \) between the compute nodes after redistribution will also satisfy \( B\le\epsilon \).
\end{theorem}

\begin{proof}
The relationship between \( \widehat{S}_{skew}^{ij} \) and \( \widehat{S}_{skew}^j \) is:

\begin{equation}
\label{equation:Sj}
	|\widehat{S}_{skew}^j| = \sum^{n-1}_{i=0} |\widehat{S}_{skew}^{ij}|
\end{equation}

Let $B_c=\epsilon$, transform \cref{equation:BalanceFactor} into the following inequality:
\begin{equation}
	{\mathop{max}\limits_{j} |\widehat{S}_{skew}^{ij}|} - {\mathop{min}\limits_{j} |\widehat{S}_{skew}^{ij}|} \leq B_c \cdot {\mathop{max}\limits_{j} |\widehat{S}_{skew}^{ij}|}
\end{equation}

Apply a summation transformation to the inequality:
\begin{equation}
	\sum^{n-1}_{i=0}({\mathop{max}\limits_{j} |\widehat{S}_{skew}^{ij}|} - {\mathop{min}\limits_{j} |\widehat{S}_{skew}^{ij}|}) \leq \sum^{n-1}_{i=0}(B_c \cdot {\mathop{max}\limits_{j} |\widehat{S}_{skew}^{ij}|})
\end{equation}

Rearrange the summation:
\begin{equation}
	{\mathop{max}\limits_{j} (\sum^{n-1}_{i=0} |\widehat{S}_{skew}^{ij}|)} - {\mathop{min}\limits_{j} (\sum^{n-1}_{i=0} |\widehat{S}_{skew}^{ij}|)} \leq B_c \cdot {\mathop{max}\limits_{j} (\sum^{n-1}_{i=0} |\widehat{S}_{skew}^{ij}|)}
\end{equation}

Divide both sides by $ {\mathop{max}\limits_{j} (\sum^{n-1}_{i=0} |\widehat{S}_{skew}^{ij}|)} $:
\begin{equation}
	\frac{{\mathop{max}\limits_{j} (\sum^{n-1}_{i=0} |\widehat{S}_{skew}^{ij}|)} - {\mathop{min}\limits_{j} (\sum^{n-1}_{i=0} |\widehat{S}_{skew}^{ij}|)}}{\mathop{max}\limits_{j} (\sum^{n-1}_{i=0} |\widehat{S}_{skew}^{ij}|)} \leq B_c 
\end{equation}

Substitute this into the inequality using \cref{equation:Sj}:
\begin{equation}
	\frac{{\mathop{max}\limits_{j} |\widehat{S}^{j}|} - {\mathop{min}\limits_{j} |\widehat{S}^{j}|}}{\mathop{max}\limits_{j} |\widehat{S}^{j}|} \leq B_c
\end{equation}

Combining this with \cref{equation:B}, we obtain:
\begin{equation}
	B \leq B_c
\end{equation}
\end{proof}

Consequently, ensuring the global requirement \( B \le \epsilon \) reduces to satisfying \( B^i \le \epsilon \) locally. 
To achieve this, we present the \textit{Balanced Partition} algorithm (\cref{algorithm:BalancedPartition}).
For each incoming skewed tuple \( t \) with join key \( \boldsymbol{x} \), the algorithm executes the following steps:

\begin{enumerate}
    \item \textbf{Tentative Assignment:} The algorithm tentatively assigns \( t \) to \( C(t') \), where \( t' \) is the previous tuple with the same key \( \boldsymbol{x} \), thereby preserving \textit{allocation locality}.
    
    \item \textbf{Balance Validation:} It then calculates the local balance factor \( B^i \). If the condition \( B^i \le \epsilon \) is satisfied, the assignment is confirmed.
    
    \item \textbf{Reassignment (Greedy Adjustment):} If \( B^i \le \epsilon \) is violated, the algorithm attempts to reassign \( t \) to the node in \( U(x) \) with the minimum current load.
    
    \item \textbf{Set Expansion:} If assigning to the minimum-load node still fails to satisfy \( B^i \le \epsilon \) (or if \( C(t') \) was already the minimum node), it indicates that the current capacity of \( U(x) \) is insufficient to accommodate the load without disrupting the overall balance. Only under this condition does \( U(x) \) expand according to \cref{algorithm:UpdateU}.
\end{enumerate}

\noindent\underline{\textit{Example:}} Assume the state in \cref{fig:Bala-Join} represents the current cluster status, where Data Node 2 has routed two skewed tuples ($\langle x, 4 \rangle_p, \langle x, 5 \rangle_p$) to Compute Node 1. The current candidate set is $U^2(x)=\{\text{Compute Node 1}\}$.
Suppose a new skewed tuple $\langle x, 7 \rangle_p$ arrives at Data Node 2.
The detector identifies it as skewed.
The algorithm tentatively assigns it to Compute Node 1 (following the previous placement).
Under this tentative assignment, Compute Node 1 would hold 3 skewed tuples ($2$ existing $+ 1$ new), while Compute Nodes 2 and 3 hold 0.
Consequently, the balance factor is calculated as $B = (3-0)/3 = 1$.
Since $B > \epsilon$, the set $U^2(x)$ is expanded to $\{\text{Compute Node 1}, \text{Compute Node 2}\}$.
Finally, $\langle x, 7 \rangle_p$ is distributed to Compute Node 2 to mitigate the imbalance.

\eat{
In \cref{algorithm:BalancedPartition}, Lines 1-15 calculate the balance factor based on \cref{equation:BalanceFactor}. The calculation of the balance factor has a linear time complexity. In practice, it can be further improved to achieve constant time complexity. As demonstrated in Line 27, each partition adds only one tuple at a time. Therefore, by keeping track of $ maxVal $, $ minVal $, and the size of each partition, we can efficiently update $ maxVal $ and $ minVal $ whenever a new tuple is added to a partition.

Lines 17-18 initialize an empty partition set and maps. Lines 19-41 assign partitions for each tuple in \( S_{skew}^i \) to achieve a balanced distribution. Lines 22-24 distribute the tuple with a newly identified skewed join key to the partition specified by the first element in the sequence. Lines 26-28 attempt to assign \( t \) to the partition previously selected for the tuple with the join key \( \boldsymbol{x} \). Lines 29-31 check the balance factor and reselect candidate nodes from \( U^i(x) \) if \( B^i > \epsilon \). If none of the nodes in \( U^i(x) \) meet the requirement (checked in Line 32), Lines 33-34 expand the \( U^i(x) \) to contain a new node. Finally, Line 42 returns the partition set \ie $ partition $.
}

\cref{algorithm:BalancedPartition} illustrates the detailed execution workflow of the proposed strategy. 
Notably, the calculation of the balance factor can be optimized to achieve constant time complexity ($O(1)$) in practice. 
By maintaining the current state of $ \max $ and $ \min $ partition sizes, the balance factor can be updated incrementally upon each tuple assignment, thereby avoiding repetitive linear scans.

\noindent\textit{D. Algorithm Complexity}

The BPPR algorithm consists of the \textit{Balanced Partition} algorithm and the multicast mechanism. The former dispatches skewed tuples from the probe table to various nodes, while the latter replicates the tuples with the same join key in the build table to corresponding nodes. In a $n$-node cluster, skewed tuples on node $i$ are denoted as ${S}_{skew}^i$, with the number of tuples represented as $m_p=|{S}_{skew}^i|$. The number of tuples from $R_{Sskew}$ is denoted as $m_b$. Define $D({S}_{skew}^i)$ as the deduplication operation, where the number of distinct skewed join keys on each node after deduplication is denoted as $k = |D({S}_{skew}^i)|$. 

In the \textit{Balanced Partition} algorithm, skewed tuples are traversed with time complexity $O(m_p)$. The optimized computation of the balance factor has a constant time complexity $O(1)$, the same as expanding the set of nodes. Therefore, the time complexity of the \textit{Balanced Partition} algorithm is $O(m_p)$. The algorithm \textit{Balanced Partition} records the $m$ tuples and $U^i(x)$ for each skewed join key $\boldsymbol{x}$, the size of which is at most $n$ and at least $1$. The sequence for each skewed join key is kept with the size $n$. Hence, the space complexity of the \textit{Balanced Partition} algorithm is $O(m_p+k\cdot n)$. Since tuples are directly dispatched in practice, the space complexity is $O(k\cdot n)$.

For multicast, each tuple with the join key $\boldsymbol{x}$ in $R_{Sskew}$ is replicated to $U^i(x)$. Then the time complexity is at least $O(m_b)$ and at most $O(m_b \cdot n)$. The space complexity is $O(1)$.

For each data node, the time and space complexity is listed in \cref{tb:complexity}.

In summary, the \textit{Balanced Partition} algorithm significantly improves cluster-wide load balance by partitioning skewed tuples at a fine granularity. 
Meanwhile, the incremental expansion strategy for the candidate set minimizes network overhead. 
Crucially, the effectiveness of this approach is agnostic to the specific data distribution, enabling dynamic adaptation to skewed workloads, thereby addressing \textbf{Chal.} \bcirclednumber{1} and \textbf{Chal.} \bcirclednumber{3} listed in \cref{subsec:observations}. 

However, while Algorithm~\ref{algorithm:UpdateU} ensures consistency across local views, correctly multicasting build tuples necessitates the logical aggregation of these local sets \( U^i(x) \) into a global target set \( U(x) \). 
In the following section, we introduce the Distributed Detector and the \textit{ASAP} mechanism. 
These components are designed to handle streaming skew detection and facilitate the precise multicast of build tuples to \( U(x) \) without heavy synchronization.

\begin{algorithm}[t]
\renewcommand{\algorithmicrequire}{\textbf{Input:}}
\renewcommand{\algorithmicensure}{\textbf{Output:}}
\caption{Node Set Update Algorithm}
\label{algorithm:UpdateU}
\begin{algorithmic}[1]
\Require join key $\boldsymbol{x}$, sequence $\boldsymbol{seq}$ for $\boldsymbol{x}$, number of nodes $\boldsymbol{n}$
\Ensure updated $\boldsymbol{U^{i}(x)}$

\Procedure{GenSeq}{$\boldsymbol{x}, \boldsymbol{seq}, \boldsymbol{n}$}
    \If{$\boldsymbol{seq}.size < \boldsymbol{n}$}
        \State $epoch \gets \boldsymbol{seq}.size$
        \State $candidate \gets hash(\boldsymbol{x} + epoch) \bmod \boldsymbol{n}$
        \While{$candidate \in \boldsymbol{seq}$}
            \State $epoch \gets epoch + 1$
            \State $candidate \gets hash(\boldsymbol{x} + epoch) \bmod \boldsymbol{n}$
        \EndWhile
        \State $\boldsymbol{seq}.add(candidate)$
    \EndIf
    \State \textbf{return} $\boldsymbol{seq}$
\EndProcedure

\Procedure{UpdateU}{$\boldsymbol{x}, \boldsymbol{U^{i}(x)}, \boldsymbol{seq}$}
    \State $size \gets \boldsymbol{U^{i}(x)}.size$
    \State $\boldsymbol{seq} \gets$ \Call{GenSeq}{$\boldsymbol{x}, \boldsymbol{seq}, \boldsymbol{n}$}
    \If{$size == \boldsymbol{seq}.size$}
        \State \textbf{return} $\boldsymbol{U^{i}(x)}$
    \EndIf
    \State $\boldsymbol{U^{i}(x)}.add(\boldsymbol{seq}[size])$
    \State \textbf{return} $\boldsymbol{U^{i}(x)}$
\EndProcedure
\end{algorithmic}
\end{algorithm}

\begin{algorithm}[!t]
\renewcommand{\algorithmicrequire}{\textbf{Input:}}
\renewcommand{\algorithmicensure}{\textbf{Output:}}
\caption{Balanced Partition Algorithm}
\label{algorithm:BalancedPartition}
\begin{algorithmic}[1]
\Require Skewed tuples $\boldsymbol{S_{skew}^i}$, Threshold $\epsilon$, Number of nodes $\boldsymbol{n}$
\Ensure Set of balanced subpartitions $\boldsymbol{partition}$

\Procedure{CalcBalanceFactor}{$\boldsymbol{partition}$}
    \State $maxVal \gets 0$
    \State $minVal \gets \infty$
    \ForAll{$sub \in \boldsymbol{partition}$}
        \State $size \gets sub.size$
        \If{$size > maxVal$} $maxVal \gets size$ \EndIf
        \If{$size < minVal$} $minVal \gets size$ \EndIf
    \EndFor
    \If{$maxVal == 0$} \State \textbf{return} $0$ \EndIf
    \State \textbf{return} $(maxVal - minVal) / maxVal$
\EndProcedure

\Procedure{BalancedPartition}{$\boldsymbol{S_{skew}^i}, \epsilon, \boldsymbol{n}$}
    \State $\boldsymbol{partition} \gets$ Initialize $\boldsymbol{n}$ empty lists
    \State $\boldsymbol{U^i}, \boldsymbol{Seq} \gets$ Initialize empty maps
    \ForAll{$t \in \boldsymbol{S_{skew}^i}$}
        \State $\boldsymbol{x} \gets t.b$
        \State $ID \gets 0$ \Comment{Node index}
        \If{$\boldsymbol{Seq}_{\boldsymbol{x}} == []$}
            \State \Call{UpdateU}{$\boldsymbol{x}, \boldsymbol{U^i(x)}, \boldsymbol{Seq}_{\boldsymbol{x}}$}
            \State $ID \gets \boldsymbol{Seq}_{\boldsymbol{x}}[0]$
        \Else
            \State $lastID \gets$ Subpartition ID last selected by $\boldsymbol{x}$
            \State $\boldsymbol{partition}[lastID].add(t)$ \Comment{Tentative add}
            \State $balanceFactor \gets$ \Call{CalcBalanceFactor}{$\boldsymbol{partition}$}
            \State $\boldsymbol{partition}[lastID].delete(t)$ \Comment{Rollback}
            \If{$balanceFactor > \epsilon$}
                \State $ID \gets$ min-size subpartition in $\boldsymbol{U^i(x)}$
                \If{$ID == lastID$}
                    \State \Call{UpdateU}{$\boldsymbol{x}, \boldsymbol{U^i(x)}, \boldsymbol{Seq}_{\boldsymbol{x}}$}
                    \State $ID \gets \boldsymbol{Seq}_{\boldsymbol{x}}.last()$
                \EndIf
            \Else
                \State $ID \gets lastID$
            \EndIf
        \EndIf
        \State $\boldsymbol{partition}[ID].add(t)$ \Comment{Final assignment}
    \EndFor
    \State \textbf{return} $\boldsymbol{partition}$
\EndProcedure
\end{algorithmic}
\end{algorithm}

\begin{table}[t]
    \centering
    \caption{Time and Space Complexity of BPPR}\label{tb:complexity}
    \resizebox{\linewidth}{!}{%
    \begin{tabular}{|c|c|c|c|}
        \hline
        & Balanced Partition & Multicast & BPPR \\
        \hline
        time & $O(m_p)$ & $O(m_b \cdot n)$ & $O(\max(m_p,m_b \cdot n))$ \\
        \hline
        space & $O(k \cdot n)$ & $O(1)$ & $O(k \cdot n)$ \\
        \hline
    \end{tabular}%
    }
\end{table}

\section{DISTRIBUTED DETECTOR ON-THE-FLY}
\label{sec:det}
{
Building upon the partitioning strategy established in \cref{sec:bala-join}, this section details the runtime components required to realize Bala-Join in a distributed streaming environment. 
We introduce two core modules: the \textbf{Distributed Skew Detector}, which identifies skewed keys on-the-fly without global synchronization, and the \textbf{ASAP} mechanism, which implements the receiver-driven synchronization of build tuples to the dynamic target sets \( U(x) \).

\subsection{Distributed Skew Detector for BPPR}
\label{ssec:detector}
Since global statistics are unavailable for data streams or intermediate results, we employ a local detection strategy based on the \textit{Space Saving} algorithm~\cite{metwally2005efficient}. 
Unlike exact counting methods that require linear space, \textit{Space Saving} efficiently identifies heavy hitters using bounded memory (\ie \( O(k) \) counters for tracking \( k \) elements). The primary challenge of the Space Saving algorithm lies in efficiently updating elements and removing elements corresponding to the minimum count. Following the suggestion of \cite{rodiger2016flow}, we adopt the Space Saving algorithm based on a hash table and an ordered array as the detector for the single-node environment, therefore achieving minimized time and space cost.

A straightforward approach~\cite{huang2014ld,cafaro2016parallel} to distributed detection is to perform local counting on each node, aggregate lists globally to reach consensus, and then redistribute tuples. However, this approach poses three critical problems:}

\textbf{Double Traversals of $S$}: Each data node must traverse the table \( S^i \) twice (once for local detection and again for redistribution). In streaming scenarios, this buffering results in significant memory and I/O overhead.

\textbf{Independent Standalone Detection}: Varying processing speeds across nodes lead to synchronization barriers. Faster nodes become idle while waiting for slower nodes to complete detection, hindering overall throughput.

\textbf{Communication Overhead for Consensus}: Achieving a globally recognized list requires unavoidable and expensive communication overhead between data nodes.

{
To overcome these limitations, Bala-Join abandons global consensus in favor of a completely local detection strategy tightly integrated with the BPPR algorithm. 
By treating the detector as a \textit{gatekeeper}, the system achieves detection and redistribution in a \textbf{single pass} over the data stream. 
Upon the arrival of a probe tuple \( t \) with join key \( \boldsymbol{x} \), the system executes the following classification logic:

\begin{itemize}
    \item \textbf{Non-Skewed Flow (Default):} If \( t \) is not identified as skewed locally, it follows the standard hash partitioning path. It is directly routed to its hash-distributed node (calculated as \( q = \text{hash}(x) \)).
    
    \item \textbf{Skewed Flow (BPPR):} If \( t \) is identified as skewed, it bypasses the default hash path and is handed over to the \textit{Balanced Partition} algorithm. The algorithm routes \( t \) to a target node within the dynamic candidate set \( U^i(x) \) to maintain the balance factor \( B^i \le \epsilon \).
\end{itemize}

Crucially, this architecture decouples detection accuracy from routing consistency. 
Due to local data variations, the join key \( \boldsymbol{x} \) might be identified as \textit{skewed} on Data Node 1 (routed to \( U^i(x) \)) but \textit{non-skewed} on Data Node 3 (routed to the hash-distributed node \( q \)). 
The robustness of our system relies on the sequence generator (\cref{algorithm:UpdateU}), which ensures that the hash-distributed node \( q \) is strictly included within the candidate set \( U^i(x) \) (\ie \( q \in U^i(x) \)).
This guarantee ensures that regardless of whether a key is treated as skewed or non-skewed locally, the build tuples residing at node \( q \) remain accessible to the entire set \( U(x) \), thereby enabling independent operation without synchronization latency.

This integration inherently implies a \textbf{Deferred Build Partitioning} strategy. 
Since skew is handled dynamically, Bala-Join treats all tuples in the build table \( R \) as non-skewed initially, hash-distributing them to their default hash nodes defined by \( \text{hash}(\boldsymbol{x}) \). 
The partitioning (replication) of \( R \) to the expanded set \( U^i(x) \) occurs only on demand, facilitated by the mechanism described in the following subsection.

\subsection{ASAP Mechanism}
\label{ssec:asap}
To guarantee the correctness of \textsf{HASH JOIN}s in a distributed setting, data nodes across the cluster must reach two forms of consensus:

\textbf{\textit{Consensus 1}}: A consensus on skewed keys for data nodes to handle skewed tuples from \( S \) in a balanced manner;

\textbf{\textit{Consensus 2}}: A consensus on the target set \( U(x) \) for multicasting corresponding build tuples.

We achieve consensus through the \textbf{Active-Signaling and Asynchronous-Pulling (ASAP)} mechanism. Build tuples are always hash-distributed as they arrive. Upon the arrival of probe tuples, non-skewed tuples follow hash distribution, while skewed tuples are routed by BPPR. The relevant nodes are notified that they have received a skewed tuple when it arrives. Then they asynchronously pull the matching build tuples via ASAP, thereby ensuring correct join results without any pre-scan or materialization of intermediate tables. Therefore, Bala-Join could complete skew detection and redistribution in a single pass.

\noindent\underline{\textit{Example:}}~\cref{fig:Bala-Join} also shows how the ASAP mechanism works. All tuples in the build table are hash-distributed. When the skewed probe tuples ${\langle x,2\rangle}_{p}$ and ${\langle x,3\rangle}_{p}$ arrive at Compute Node 2, the Compute Node is notified that it has received tuples with the join key $x$. The Compute Node then asynchronously pulls the matching build tuples ${\langle x,1\rangle}_{b} \sim {\langle x,3\rangle}_{b}$, which have been hash-distributed to Compute Node 1.

By enabling independent data pulling, the set of nodes receiving skewed probe tuples implicitly constitutes the global target set $U(x)$, to which corresponding build tuples are pulled, effectively achieving the multicast of the build data.

As established in \cref{ssec:detector}, local detection biases may cause a key \( \boldsymbol{x} \) to be identified as \textit{skewed} on some nodes but \textit{non-skewed} on others. ASAP handles this inconsistency seamlessly.
If a tuple is routed to \( q \) (false negative), \( q \) already holds the build tuples locally (received via standard hash distribution), so the join proceeds immediately without pulling.
If a tuple is routed to any other node in \( U(x) \) (true positive), the ASAP mechanism ensures it can pull the build tuples from \( q \).
Since the sequence generator guarantees \( q \in U(x) \), the build target $q$ remains within the global candidate set, ensuring that local detection errors affect only the \textit{efficiency} (load balance), but never the \textit{correctness} of the result.

In summary, the proposed detector operates purely on local state, effectively eliminating the communication overhead and synchronization barriers associated with global agreement. By leveraging the sequence generator to ensure algorithmic robustness, the system maintains guaranteed correctness even in the presence of inconsistent local views. This design explicitly resolves \textbf{Chal.} \bcirclednumber{2}, as it removes the dependency on global state coordination while ensuring that skewed tuples are handled within the valid candidate set.
}

\section{EVALUATION}
\label{sec:exp}
\subsection{Experiment Setup}
Our setup replicates real-world cross-region topologies from our industrial projects with Inspur using a purpose-built simulation framework for distributed hash joins. This establishes a realistic testbed for evaluating performance over a WAN, allowing us to assess the solution's robustness by systematically varying the network bandwidth from congested to high-performance levels.
We adopt $1\sim 8$ Huawei cloud servers deployed across Beijing, Shanghai and Guiyang respectively, which together form a cluster of 3 to 24 nodes. Each node is equipped with a 4-core CPU, 8GB of RAM, a general-purpose SSD, and operates on the Ubuntu 18.04 operating system.

We evaluate the performance of various redistribution strategies specifically designed for skewed data, including \textsf{GraHJ}~\cite{kitsuregawa1983application}, PRPD~\cite{xu2008handling}, PnR~\cite{yang2023one}, SFR~\cite{stamos1993symmetric}, and BPPR. Afterward, we compare the performance of Bala-Join with other distributed hash join baselines that employ a skew detector, as listed below.
\begin{itemize}
    \item \textsf{GraHJ}: A strategy that only uses hash partitioning
    \item \textsf{PRPD}: A strategy that keeps skewed tuples locally
    \item \textsf{PnR}: A strategy that distributes skewed tuples evenly
    \item \textsf{SFR}: A strategy that symmetrically replicates skewed tuples to nodes on rows and columns across the matrix
     \item \textsf{Flow-Join}: A comprehensive solution with the SFR strategy and a detector based on the Space Saving algorithm
\end{itemize}

For experimental data, we use a synthetic dataset with Zipf~\cite{zipf1949human} distribution, which is widely used in previous works\cite{huang2014ld,rodiger2016flow,kitsuregawa1990bucket,shatdal1993using,rodiger2014locality}. The degree of skewness of the data exhibits a positive correlation with the Zipf factor $z$. In addition, we adopt the SSB-skew dataset~\cite{justen2024polar} with the scale factor set to 10, which is a variation of the Star Schema Benchmark~\cite{o2007star}. 
We focus on the join operations in the queries and execute them through the distributed hash join.

Performance is measured by throughput and network traffic. Throughput refers to the number of result tuples generated per second by the cluster, computed as the total number of join result tuples divided by the end-to-end elapsed time (including skew detection when enabled, redistribution, and local join computation). Network traffic represents the amount of data (in bytes) transmitted over the network. The former reflects the processing efficiency of the system, while the latter indicates the network overhead.

\subsection{Redistribution Strategy}\label{ssec52}

\begin{figure}[t]
    
    \subfloat[Throughput vs. bandwidth]{
    		\label{fig:phase1-bandwidth_t}
        \includegraphics[width=0.47\linewidth]{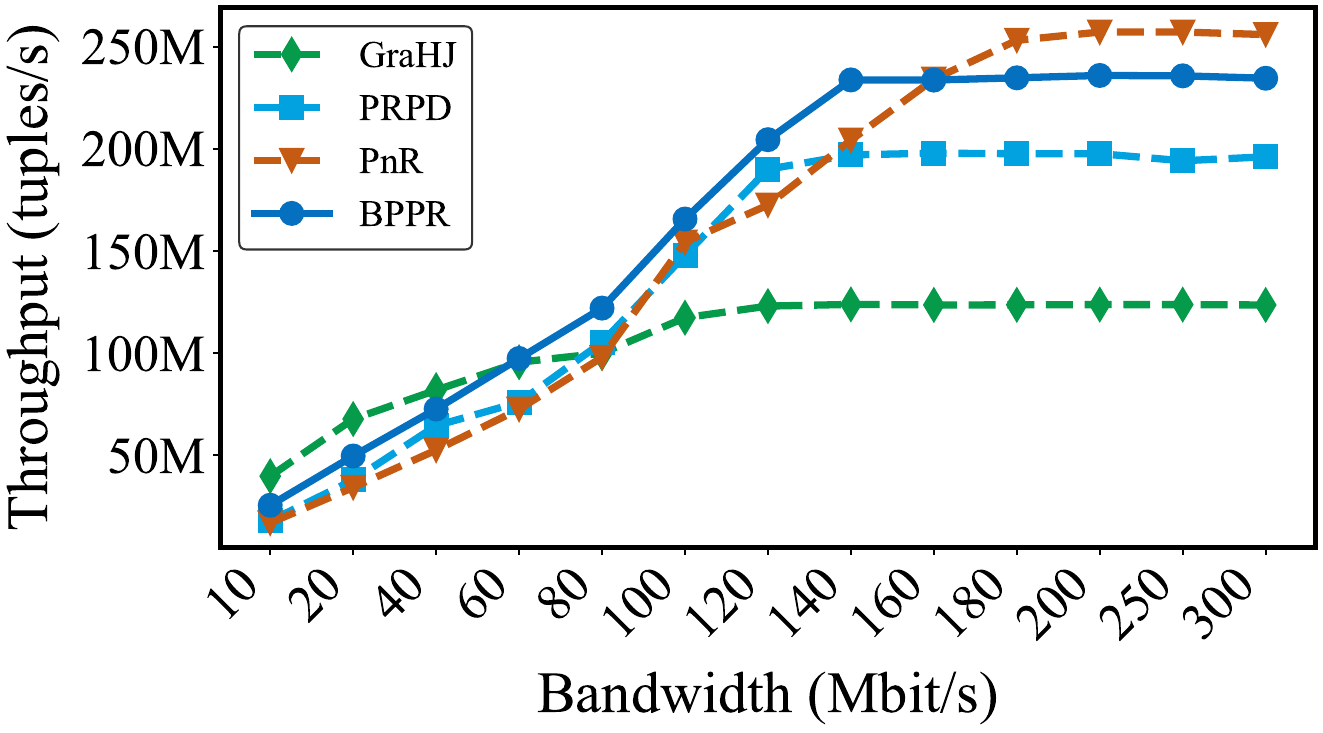}
        }
    \subfloat[Network overhead vs. bandwidth]{
    		\label{fig:phase1-bandwidth_n}
        \includegraphics[width=0.47\linewidth]{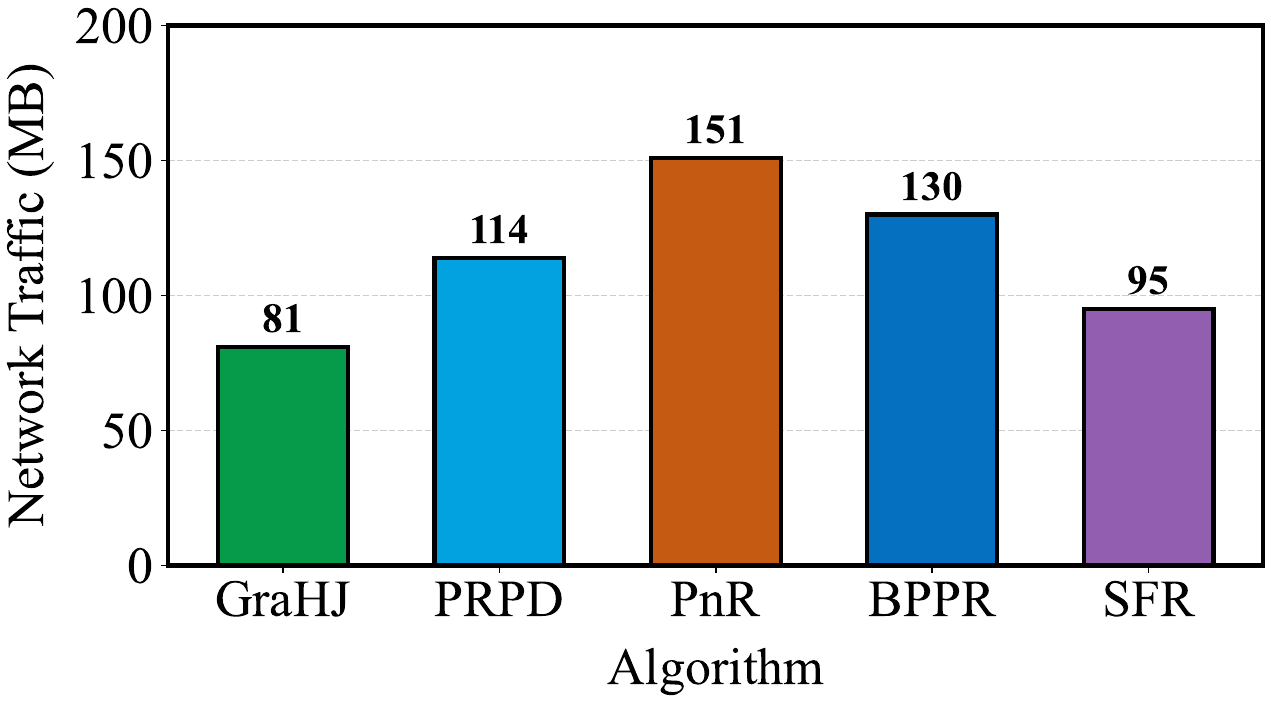}
        }\\
    \subfloat[Throughput vs. balance factor]{
    		\label{fig:phase1-balance-t}
        \includegraphics[width=0.48\linewidth]{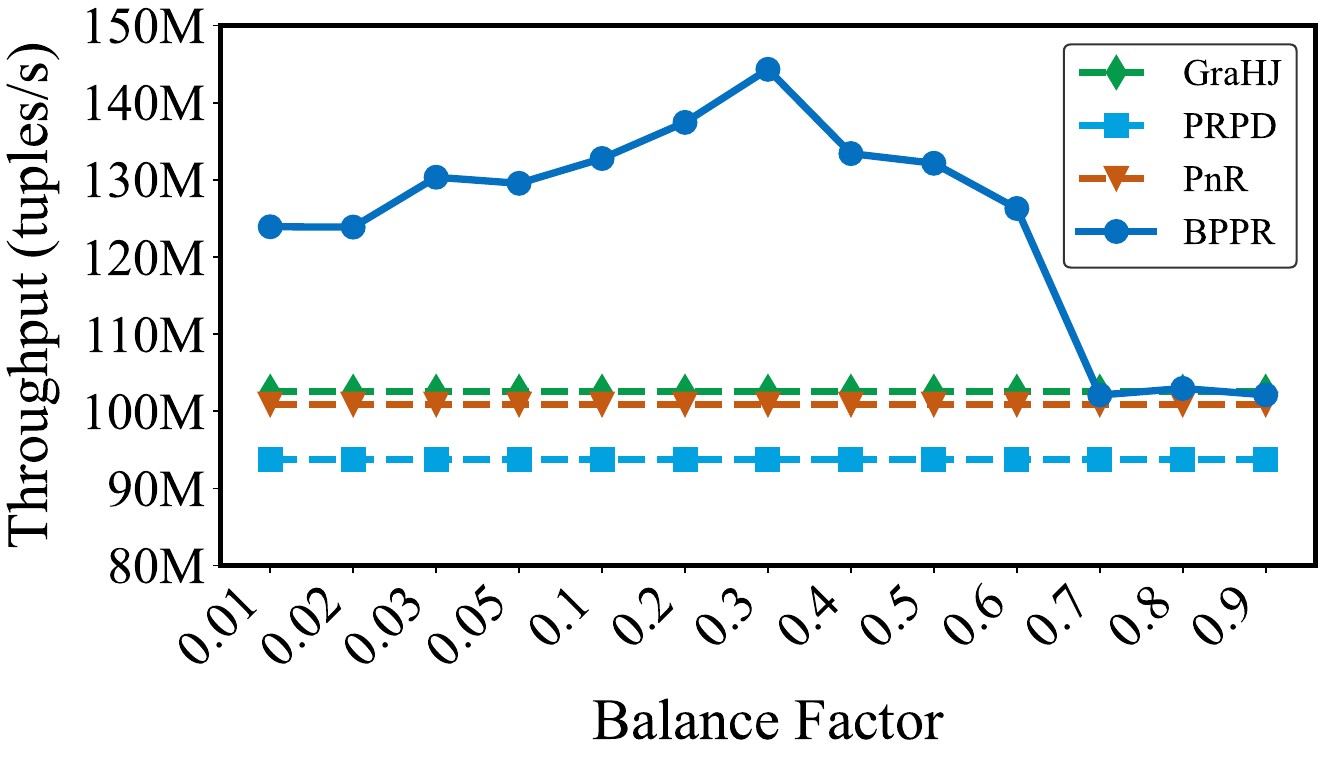}
        }
    \subfloat[Network overhead vs. balance factor]{
    		\label{fig:phase1-balance-n}
        \includegraphics[width=0.48\linewidth]{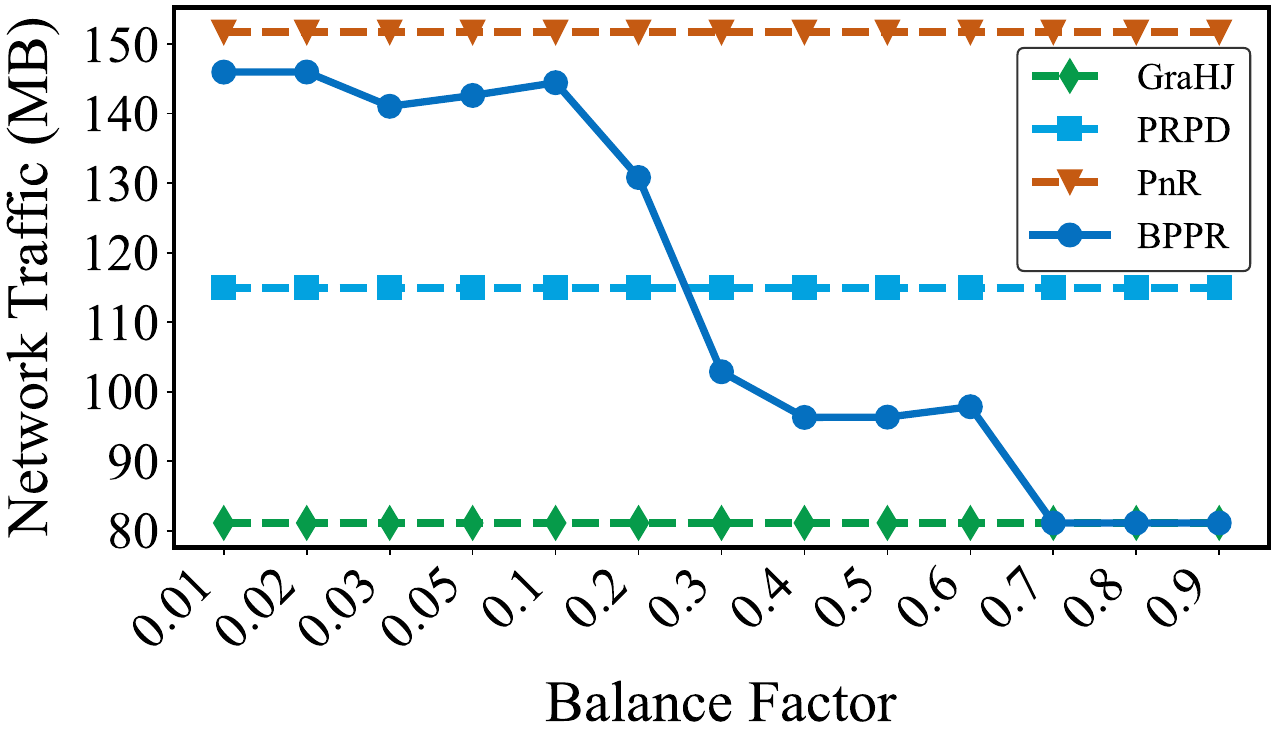}
        }\\
    \subfloat[Throughput vs. Zipf factor]{
    		\label{fig:phase1-zipf-t}
        \includegraphics[width=0.48\linewidth]{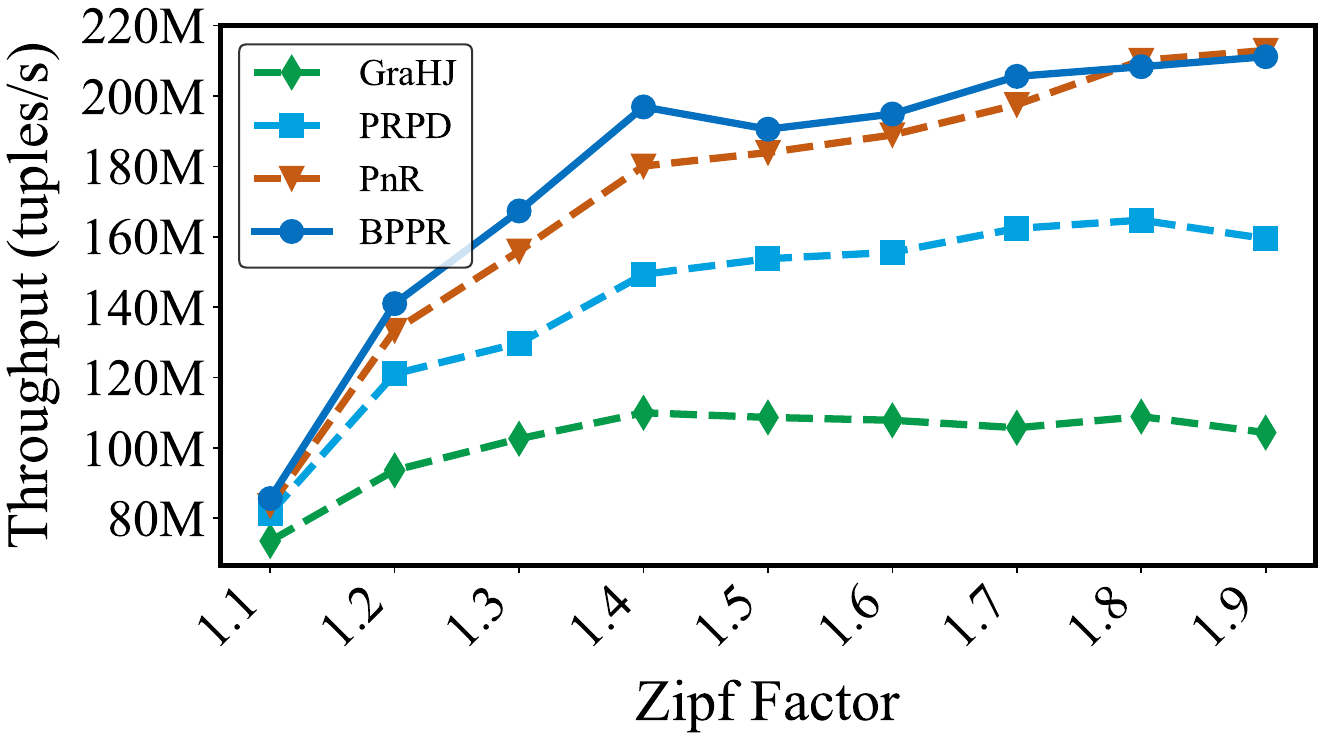}
        }
    \subfloat[Network overhead vs. Zipf factor]{
    		\label{fig:phase1-zipf-n}
        \includegraphics[width=0.48\linewidth]{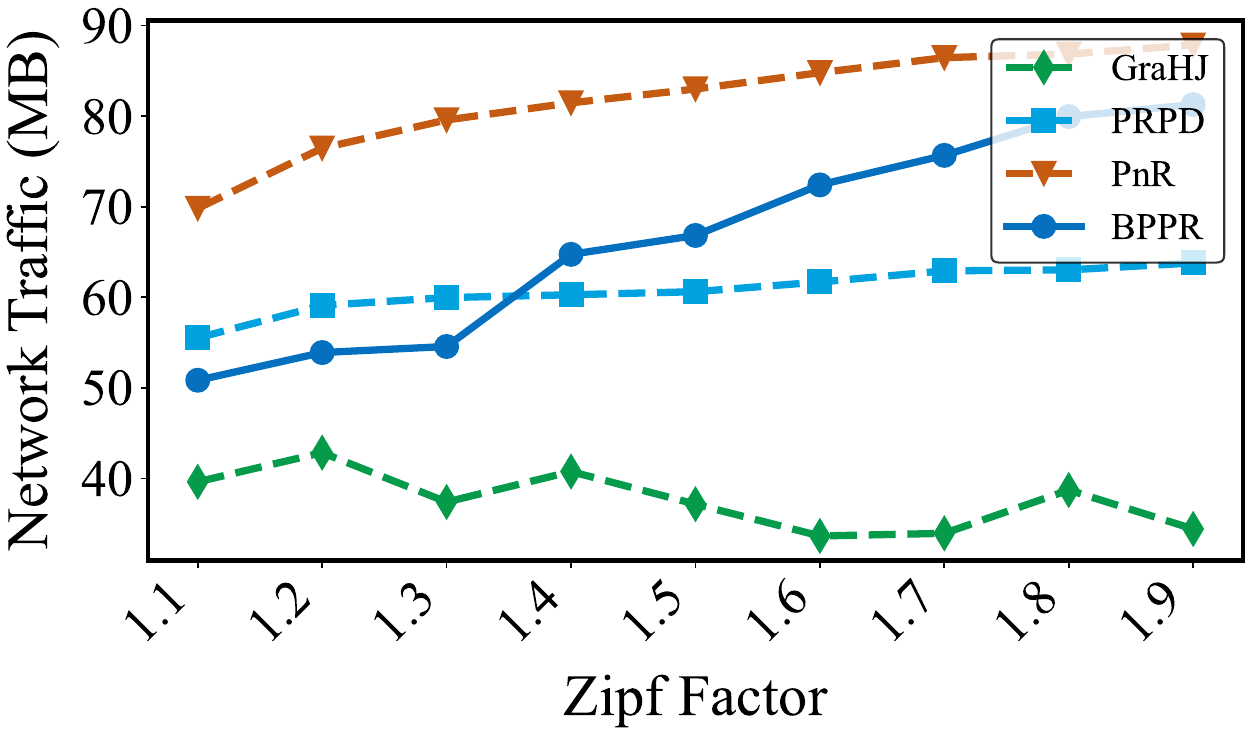}
        }\\
    \subfloat[Throughput vs. $ |R|/|S| $]{
    		\label{fig:phase1-rsratio-t}
        \includegraphics[width=0.48\linewidth]{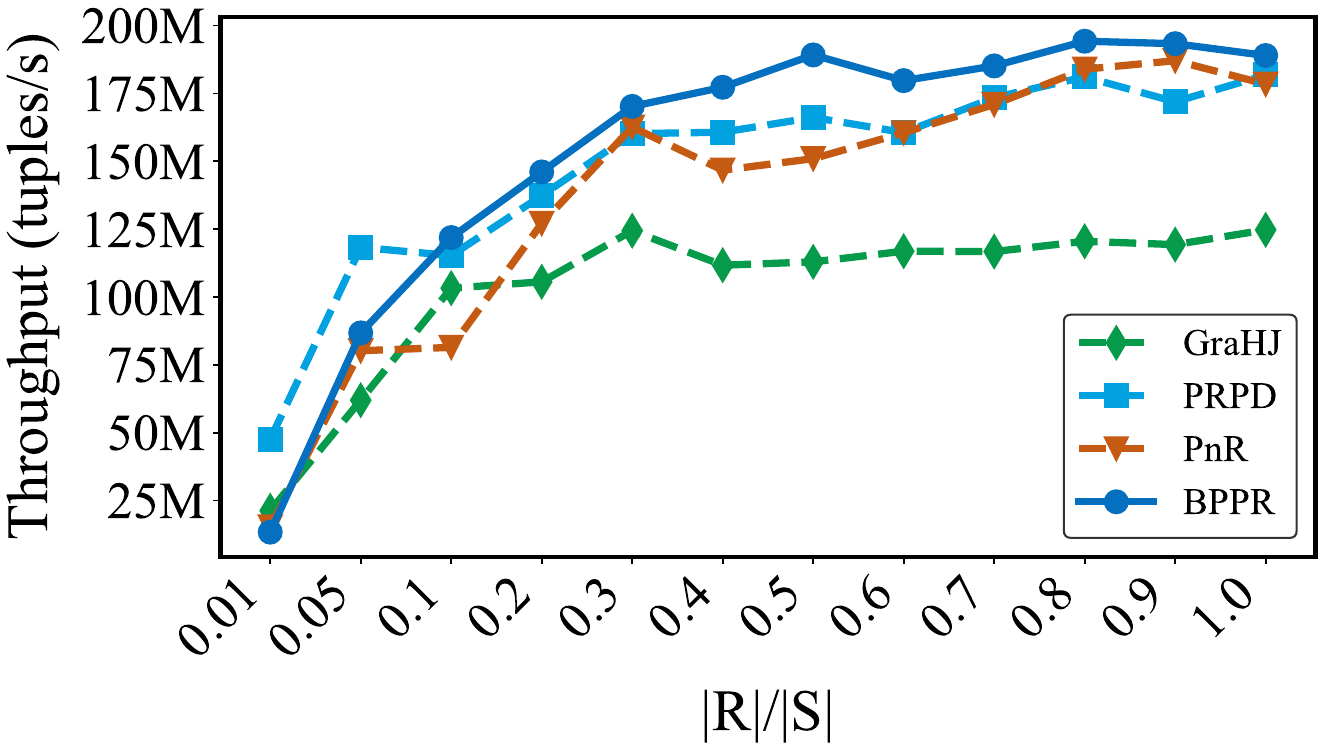}
        }   
    \subfloat[Network overhead vs. $ |R|/|S| $]{
    		\label{fig:phase1-rsratio-n}
        \includegraphics[width=0.48\linewidth]{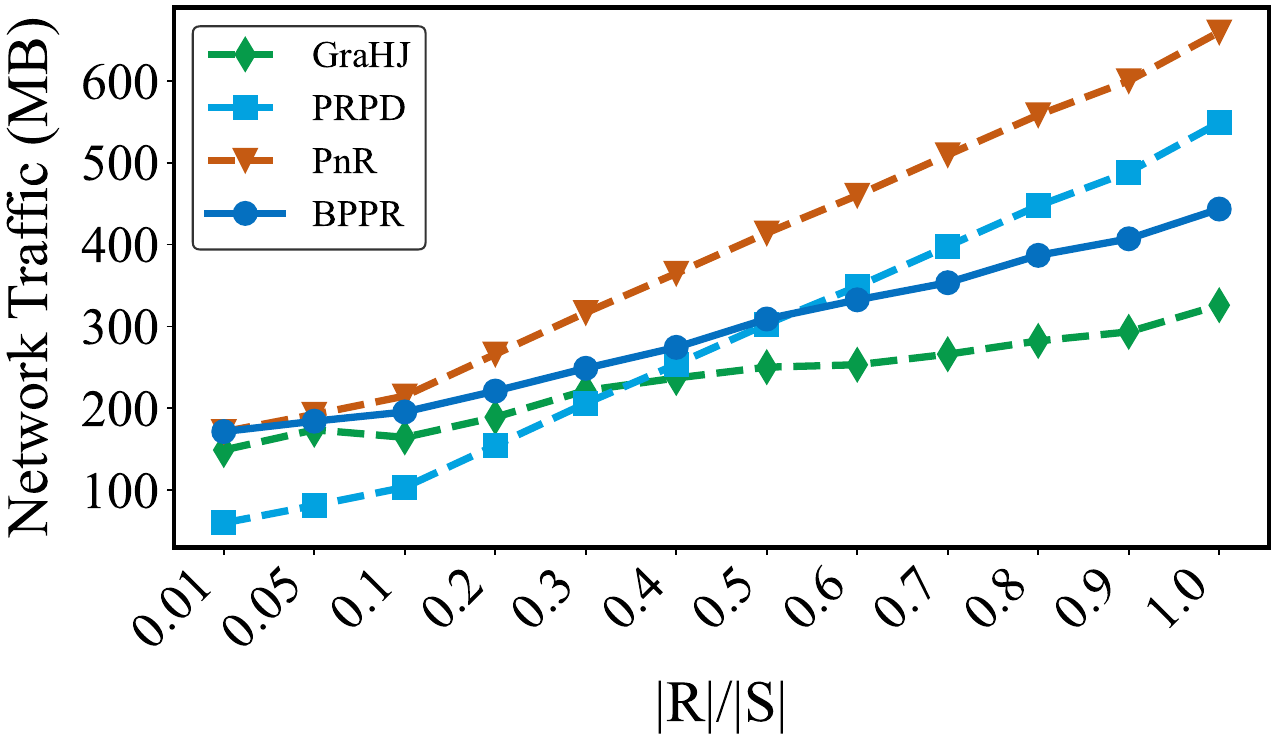}
        }\\
    \caption{Comparison of distribution strategies in various scenarios}
    \label{fig:phase1}
\end{figure}

\begin{table}
\centering
\caption{Average throughput and ranking of different algorithms over all experiments in \cref{ssec52}}
\label{table:exp1}
\resizebox{\linewidth}{!}{
\begin{tabular}{c|c|c|c|c|c}
\hline
 & Bandwidth & Balance Factor & Zipf Factor & $ |R|/|S| $ & Comprehensive Ranking \\\hline
BPPR & $ \mathbf{165M}\,(1) $ & $ \mathbf{125M}\,(1) $ & $ \mathbf{178M}(1) $ &  $ \mathbf{154M}\,(1) $ & $ \mathbf{622M}\,(1)$\\
PnR & $ 159M\,(2) $ & $ 101M\,(3) $ & $ 172M\,(2) $ &  $ 137M\,(4) $& $ 569M\,(2)$\\
PRPD & $ 140M\,(3) $ & $ 94M\,(4) $ & $ 142M\,(3) $ &  $ 148M\,(2) $& $ 524M\,(3)$\\
SFR & $ 127M\,(4) $ & $ 81M\,(5) $ & $ 131M\,(4) $ &  $ 140M\,(3) $& $ 470M\,(4)$\\ 
GraHJ & $ 105M\,(5) $ & $ 103M\,(2) $ & $ 102M\,(5) $ &  $ 103M\,(5) $& $ 413M\,(5)$\\
\hline
\end{tabular}
}
\end{table}

First, we conduct a series of experiments to evaluate the redistribution strategies. In these experiments, each algorithm is provided with the skew information directly. Default settings: bandwidth = 100 Mbit/s, balance factor = 0.2, Zipf factor = 1.25, and $|R|/|S|$ = 2/3.

\cref{fig:phase1-bandwidth_t} and \cref{fig:phase1-bandwidth_n} show the performance of five strategies at different bandwidths of 10-300 Mbit/s. GraHJ has the lowest network traffic, while BPPR's traffic is 21\% less than PnR but 14\% more than PRPD and exceeds SFR. BPPR performs best at medium bandwidths and is competitive at high bandwidths, especially when compared to GraHJ.

\cref{fig:phase1-balance-t} and \cref{fig:phase1-balance-n} explore the impact of the BPPR balance factor. As the balance factor increases, the throughput of BPPR increases in the beginning and reaches the highest value at 0.3. After that, the throughput drops and eventually becomes similar to GraHJ at 0.7. As shown in \cref{fig:phase1-balance-n}, a looser balance threshold allows BPPR to keep the node set $U(x)$ for each skewed key smaller, thereby reducing the multicast scope on the build side and significantly lowering the overall network transmission overhead.

\cref{fig:phase1-zipf-t} and \cref{fig:phase1-zipf-n} compare algorithms under different Zipf factors, which control the degree of data skew. BPPR outperforms all other algorithms in the aspect of throughput and ranks $2^{nd}$ in the aspect of network overhead.

\cref{fig:phase1-rsratio-t} and \cref{fig:phase1-rsratio-n} evaluate performance in different ratios $|R|/|S|$. BPPR remains stable, performing slightly worse than PRPD when $|R|/|S|$ is less than 0.1.

\cref{table:exp1} summarizes the average throughput and rankings in the four sets of experiments. BPPR demonstrates greater adaptability compared to PRPD, SFR, and PnR, consistently achieving the highest rank in all experimental settings.

\subsection{Distributed Skew Detector}\label{ssec53}
\begin{figure*}[t]
    \centering
    \subfloat[Throughput vs. bandwidth (3 nodes)]{
    		\label{fig:phase2-bandwidth3}
        \includegraphics[width=0.28\linewidth]{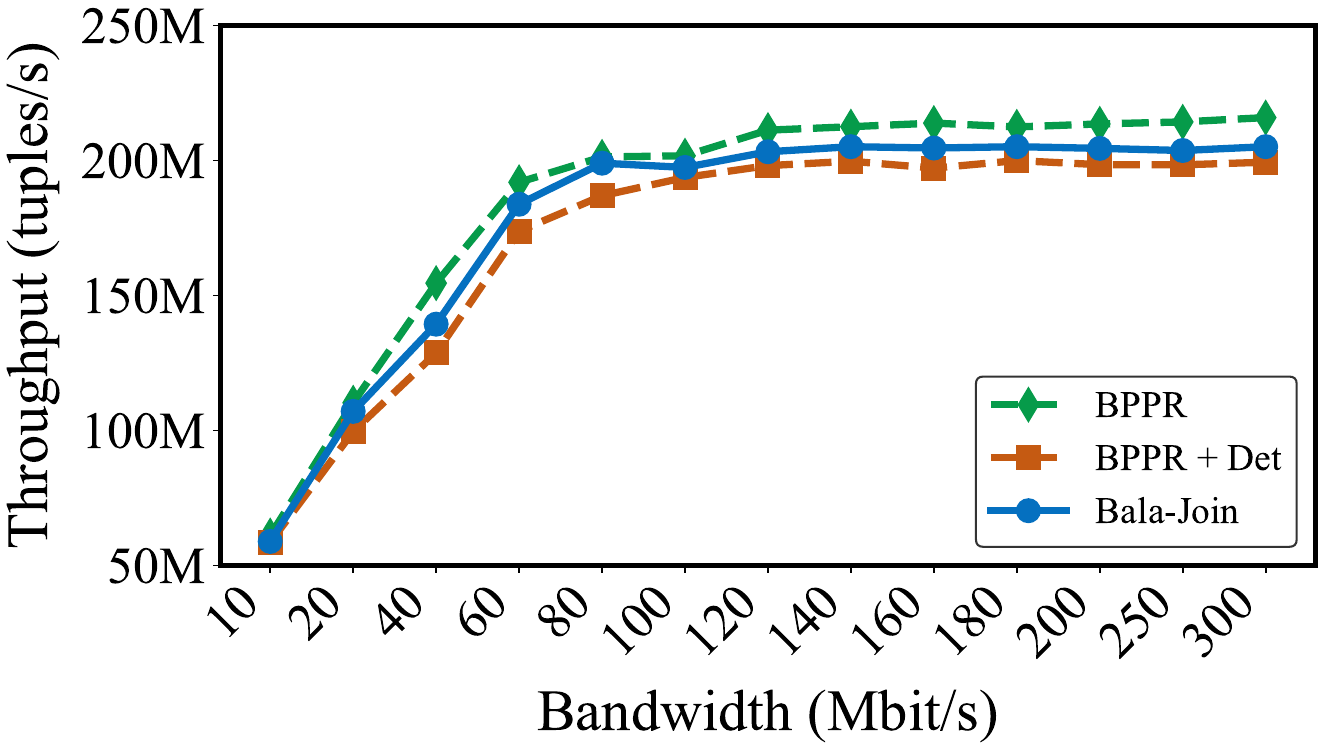}
        }\hfill
    \subfloat[Throughput vs. $ |R|/|S| $ (3 nodes)]{
    		\label{fig:phase2-rsratio3}
        \includegraphics[width=0.28\linewidth]{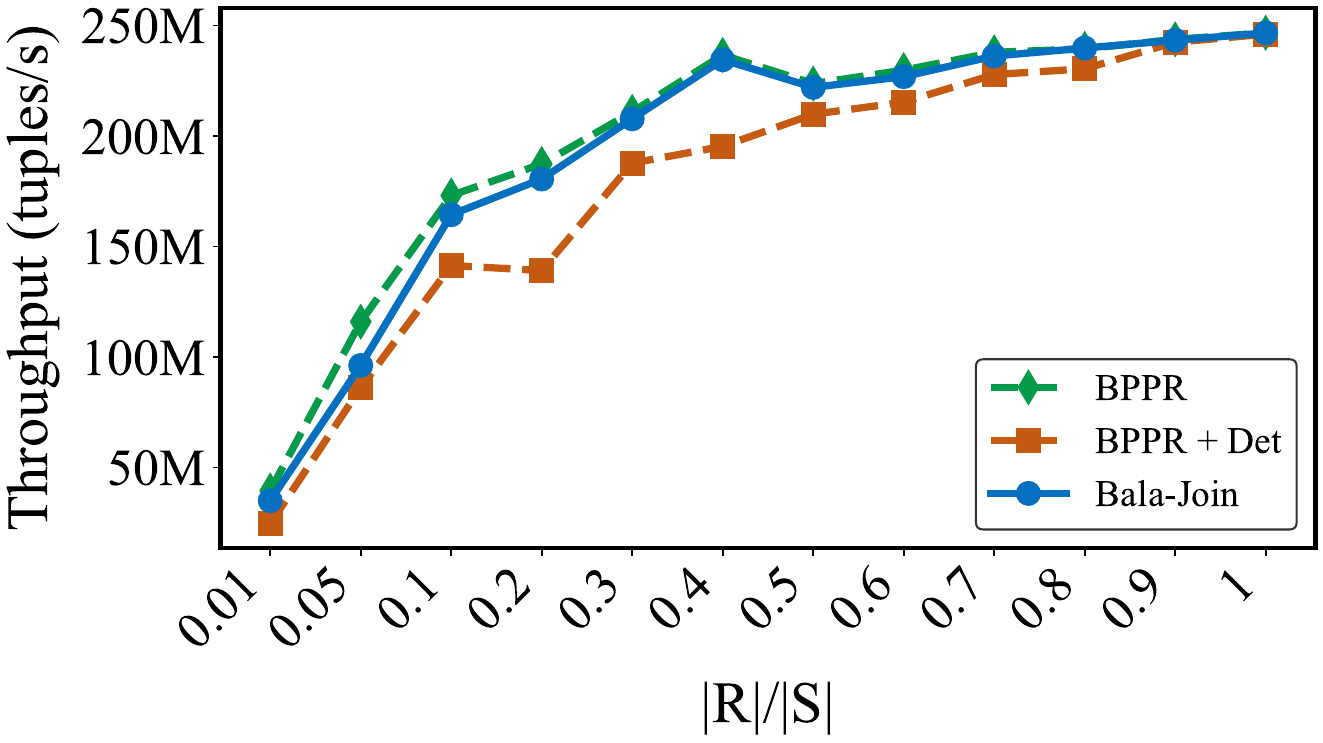}
        }\hfill
    \subfloat[Throughput vs. Zipf factor (3 nodes)]{
    		\label{fig:phase2-zipf3}
        \includegraphics[width=0.28\linewidth]{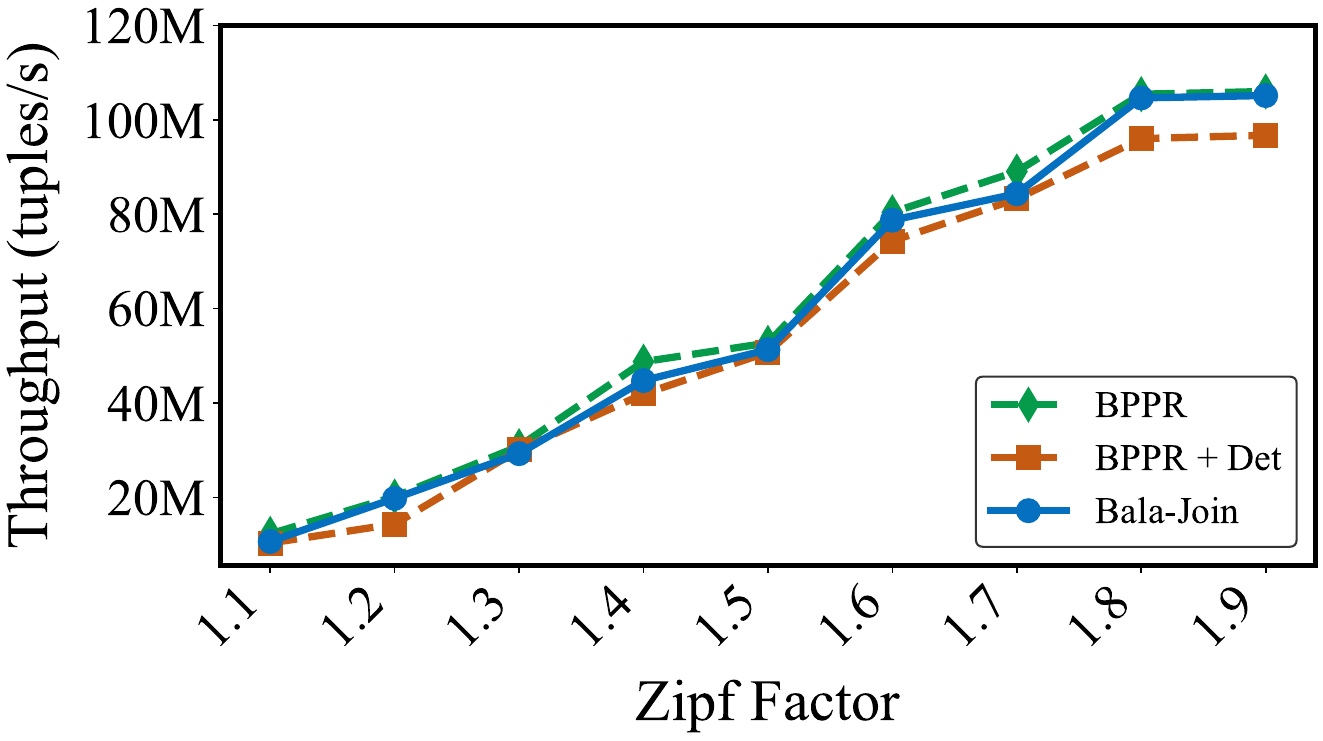}
        }\\\vspace{0ex}
    \subfloat[Throughput vs. bandwidth (6 nodes)]{
    		\label{fig:phase2-bandwidth6}
        \includegraphics[width=0.28\linewidth]{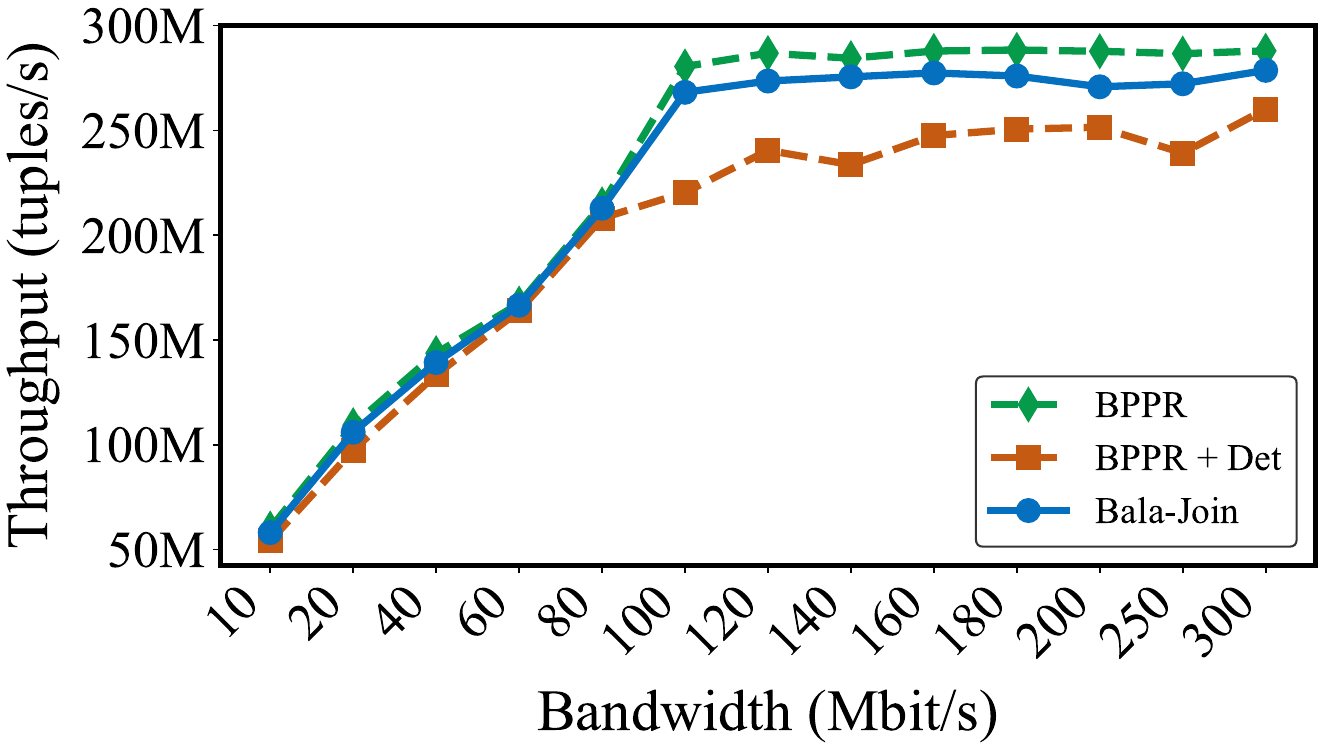}
        }\hfill
    \subfloat[Throughput vs. $ |R|/|S| $ (6 nodes)]{
    		\label{fig:phase2-rsratio6}
        \includegraphics[width=0.28\linewidth]{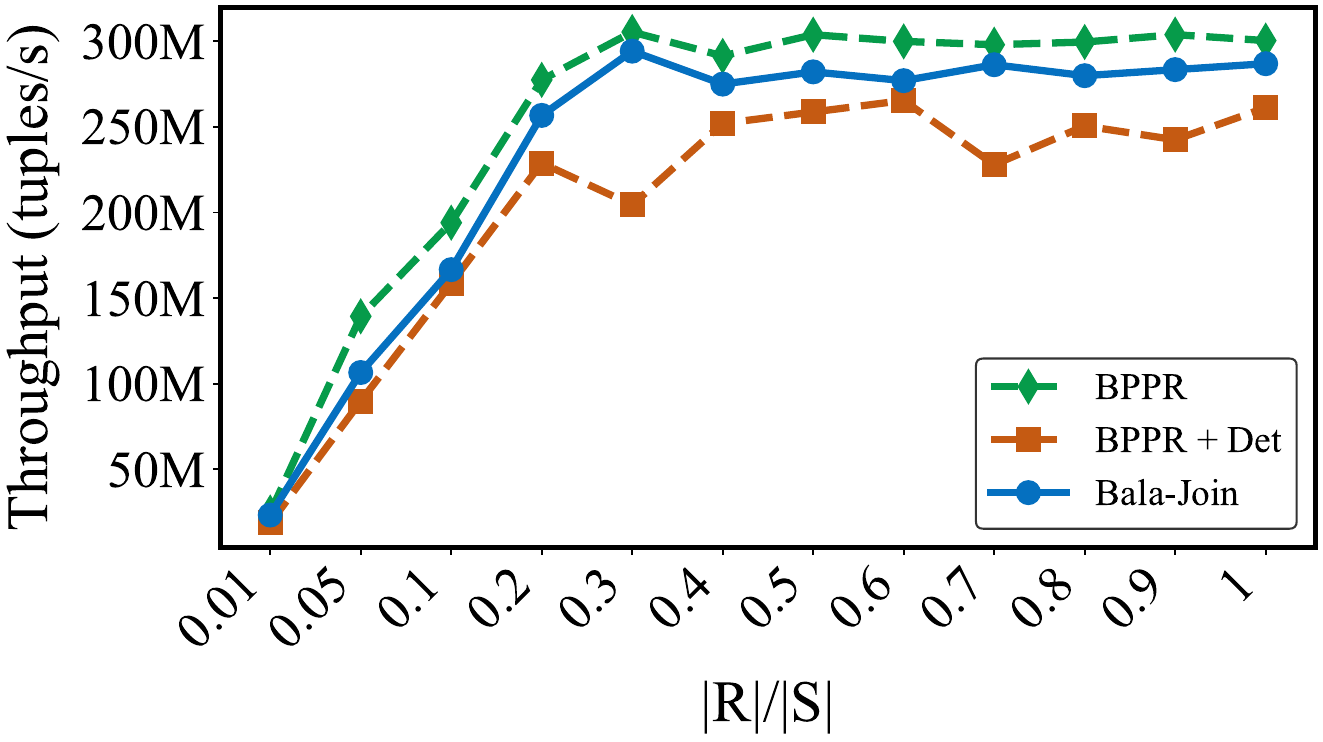}
        }\hfill
    \subfloat[Throughput vs. Zipf factor (6 nodes)]{
    		\label{fig:phase2-zipf6}
        \includegraphics[width=0.28\linewidth]{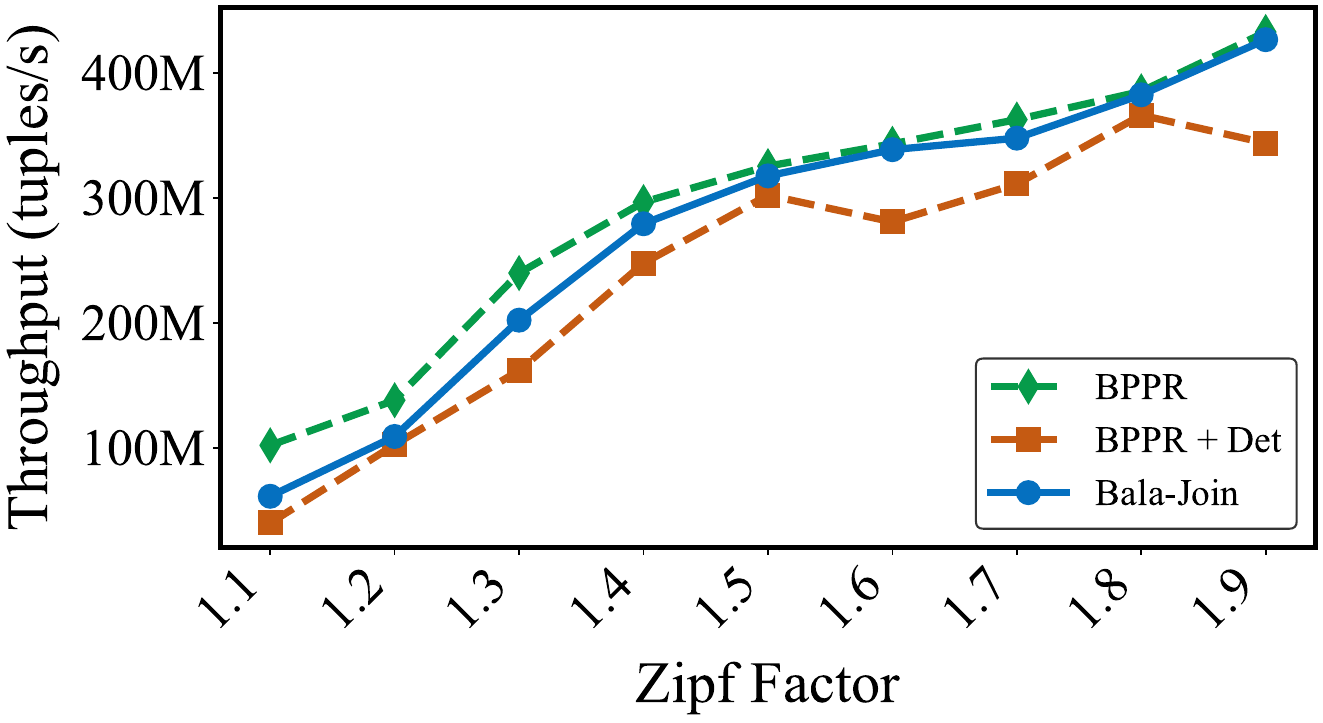}
        }\\
    \caption{Cost of the runtime skew value detector}
    \label{fig:phase2}\vspace{0ex}
\end{figure*}

In this section, we evaluate the performance of the distributed data skew detector proposed in \cref{sec:det}. Since this strategy is tightly integrated with BPPR, forming Bala-Join, we compare it with two other schemes for fairness:

\eat{\textbf{BPPR (without skew detector)}: This version directly provides skewed keys to the BPPR, allowing us to observe any overhead introduced by the distributed runtime detector.}
\textbf{BPPR (without skew detector)}: This version provides skewed values to BPPR in advance to observe any overhead introduced by the distributed runtime detector, therefore, removing the detector eliminates detection cost and leads to higher throughput.

\textbf{BPPR with an independent data skew detector (BPPR+Det)}: This employs the independent Space Saving detection algorithm \cite{cafaro2016parallel} alongside BPPR.

By comparing these three schemes, we can evaluate the overhead introduced by the distributed detector. Since all systems use BPPR for skewed data processing, there is no difference in network overhead, and the comparison focuses solely on throughput. For clarity, we abbreviate the schemes as BPPR, BPPR+Det, and Bala-Join (the solution introduced in this paper) in subsequent experiments.

\cref{fig:phase2-bandwidth3} and \cref{fig:phase2-bandwidth6} show the throughput of different strategies in 3-node and 6-node clusters across various bandwidths. In both cases, Bala-Join performs between BPPR and BPPR+Det. We also examine performance across different Zipf factors (\cref{fig:phase2-zipf3} and \cref{fig:phase2-zipf6}) and $ |R|/|S| $ ratios (\cref{fig:phase2-rsratio3} and \cref{fig:phase2-rsratio6}). 

In all experiments, the Bala-Join curve consistently falls between BPPR and BPPR+Det, indicating that the distributed detector introduces some overhead but performs significantly better than the independent Space-Saving detector used in BPPR+Det. This supports the challenges discussed in \cref{sec:det}. On average, the proposed detector adds about 5\% overhead compared to BPPR, with the 6-node setup generally incurring more overhead than the 3-node setup.

\cref{table:exp2} summarizes the average throughput per second for the three strategies across all scenarios.

\begin{table}[!t]
\centering
\caption{Average throughput and comparison of different algorithms over all experiments in \cref{ssec53}}
\label{table:exp2}
\small
\resizebox{0.90\linewidth}{!}{
\begin{tabular}{c|c|c|c}
\hline
 & Bala-Join & BPPR+Det & BPPR \\\hline
Bandwidth (3 nodes) & $ 178M $ & $ 172M(+3\%) $ & $ 186M(-4\%) $ \\
Zipf factor (3 nodes) & $ 59M $ & $ 55M(+7\%) $ & $ 61M(-8\%) $ \\
$ |R|/|S| $ (3 nodes) & $ 194M $ & $ 179M(+8\%) $ & $ 199M(-3\%) $\\
Bandwidth (6 nodes) & $ 221M $ & $ 200M(+11\%) $ & $ 230M(-4\%) $ \\
Zipf factor (6 nodes) & $ 274M $ & $ 240M(+14\%) $ & $ 292M(-6\%) $ \\
$ |R|/|S| $ (6 nodes) & $ 235M $ & $ 205M(+15\%) $ & $ 253M(-7\%) $\\
\hline
\end{tabular}
}
\end{table}

\subsection{The overall Dist-HJ Solution}\label{ssec54}

\begin{figure*}
    \centering
    \subfloat[Throughput vs. Zipf factor]{
    		\label{fig:phase3-zipf-t3}
        \includegraphics[width=0.22\linewidth]{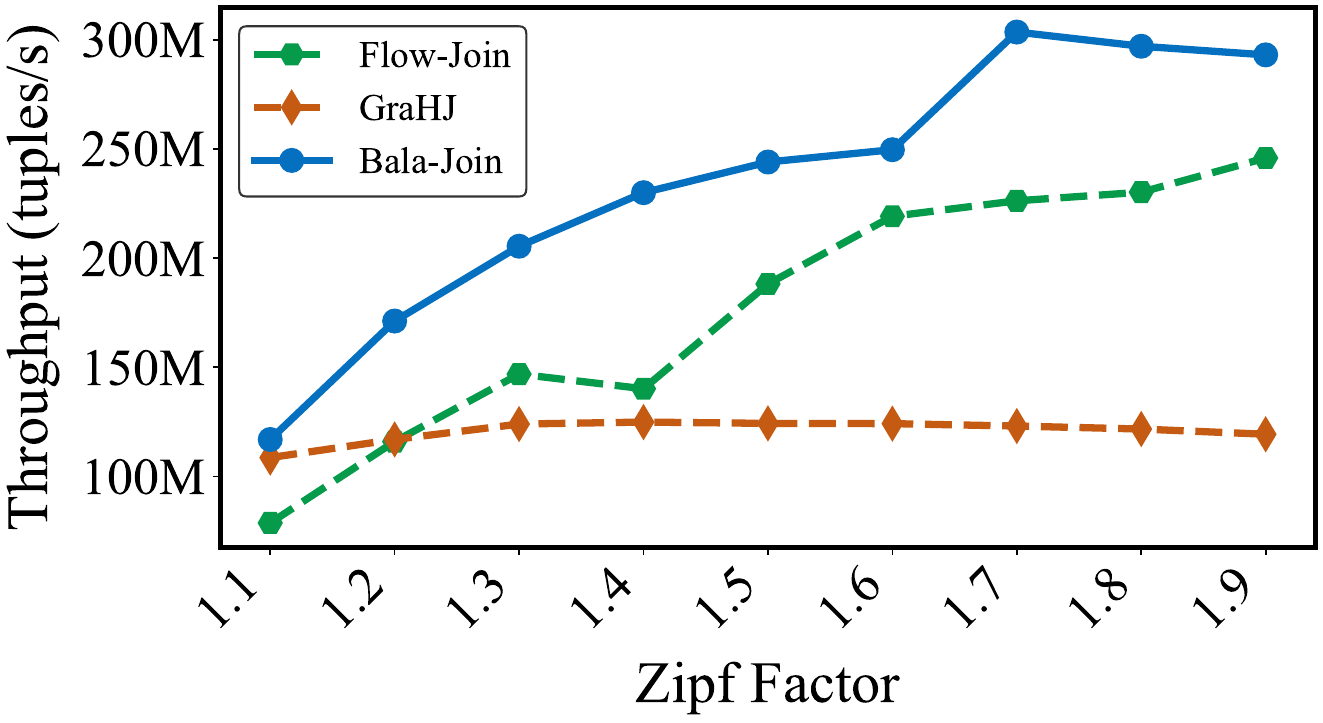}
        }\hfill
    \subfloat[Network overhead vs. Zipf factor]{
    		\label{fig:phase3-zipf-n3}
        \includegraphics[width=0.22\linewidth]{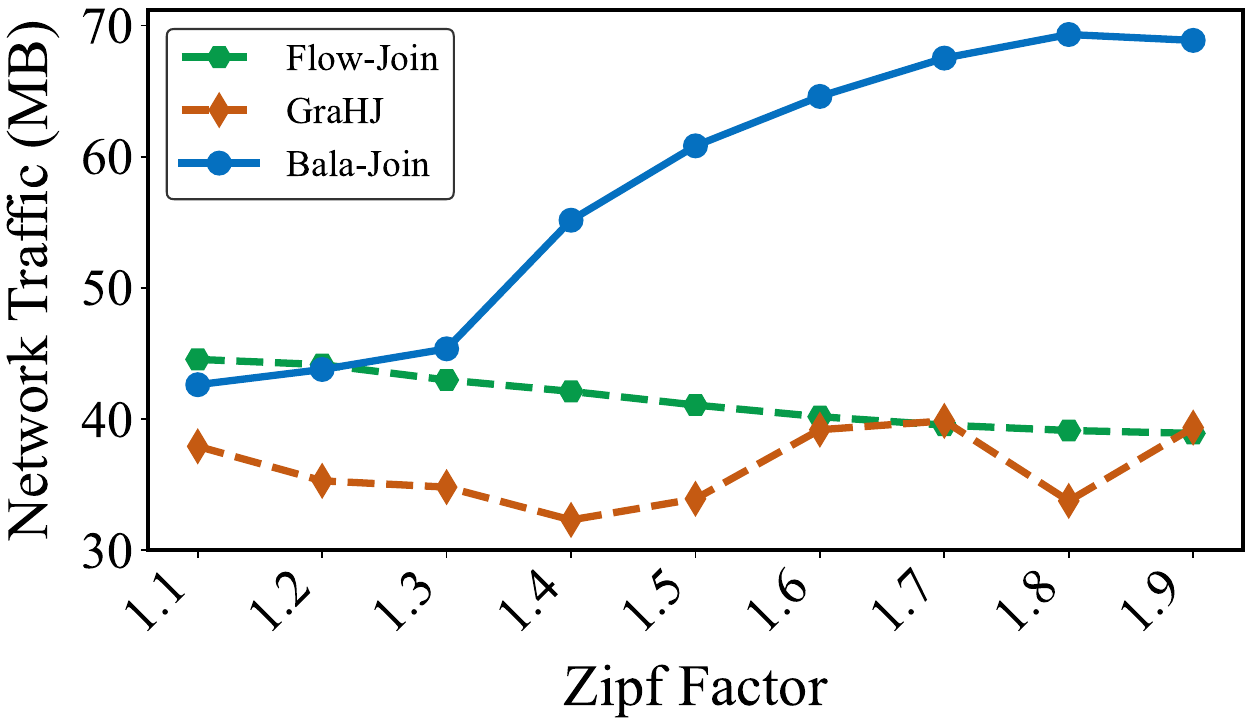}
        }\hfill
    \subfloat[Throughput vs $ |R|/|S| $]{
    		\label{fig:phase3-rsratio-t3}
        \includegraphics[width=0.22\linewidth]{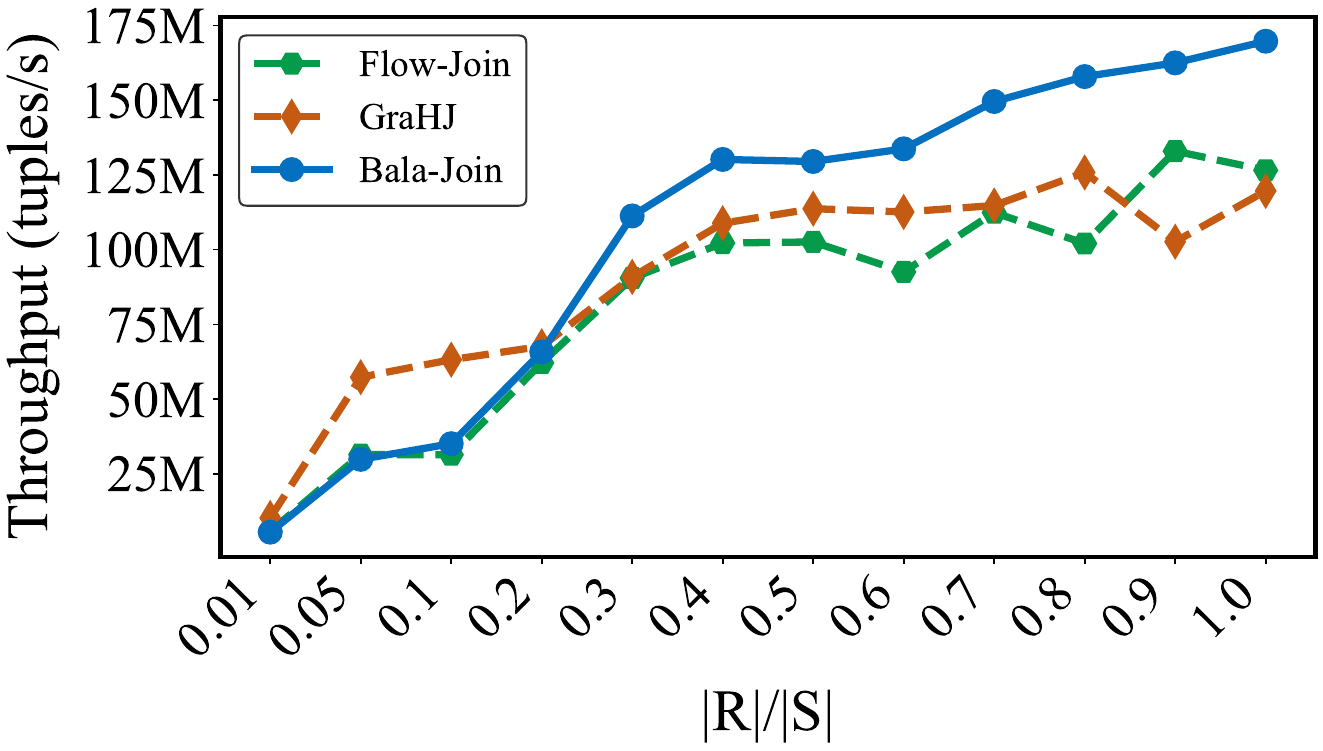}
        }\hfill
    \subfloat[Network overhead vs.  $ |R|/|S| $]{
    		\label{fig:phase3-rsratio-n3}
        \includegraphics[width=0.22\linewidth]{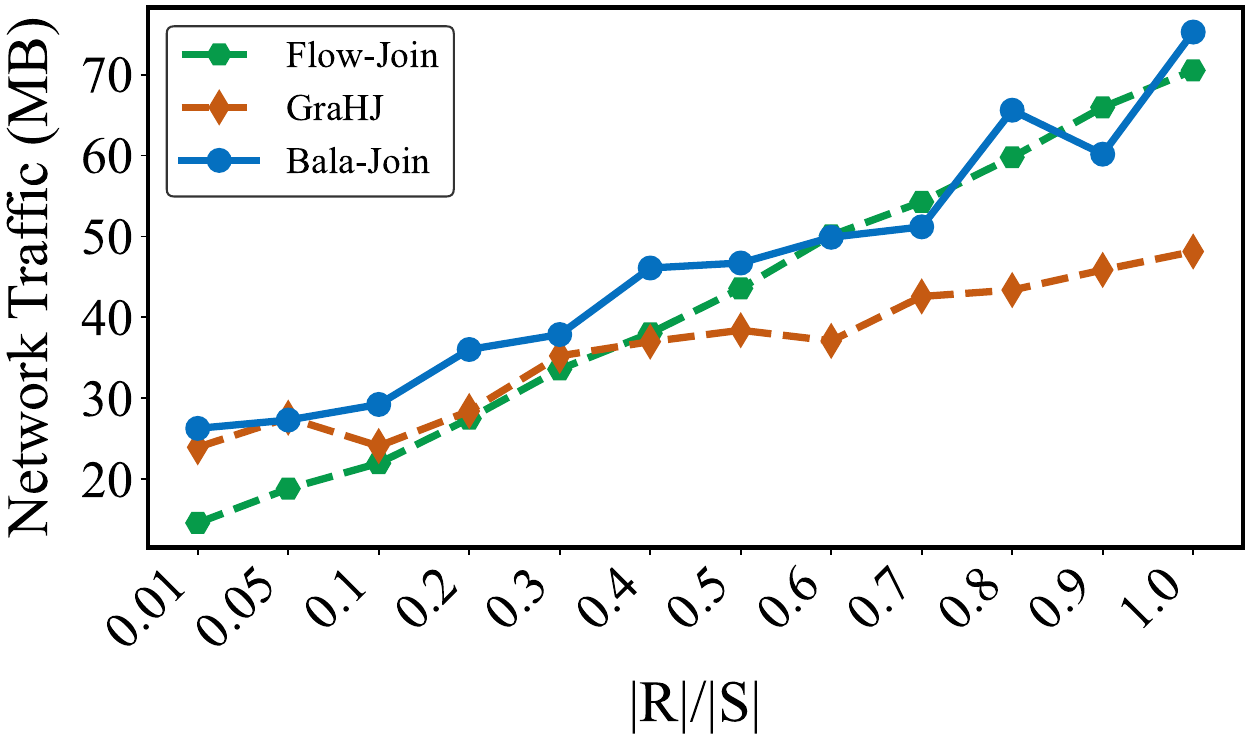}
        }\hfill 
    \subfloat[Throughput vs. Zipf factor]{
    		\label{fig:phase3-zipf-t6}
        \includegraphics[width=0.22\linewidth]{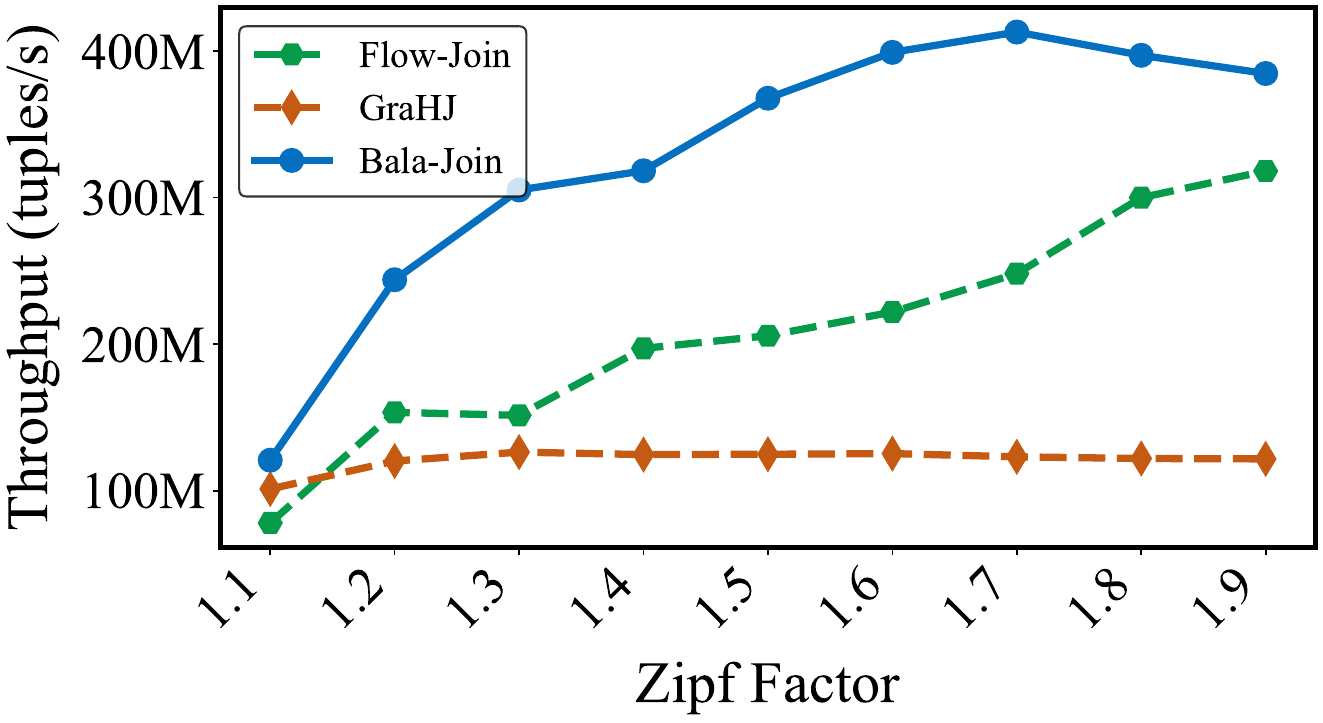}
        }\hfill
    \subfloat[Network overhead vs. Zipf factor]{
    		\label{fig:phase3-zipf-n6}
        \includegraphics[width=0.22\linewidth]{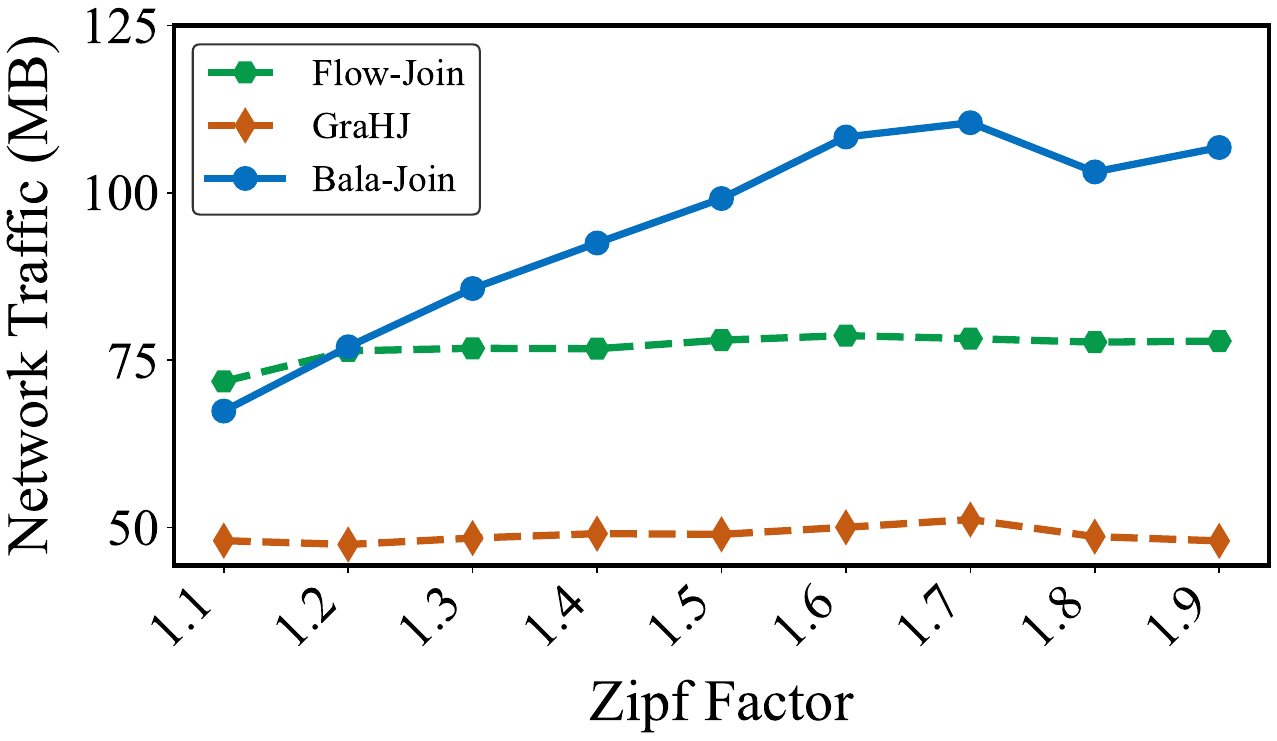}
        }\hfill
    \subfloat[Throughput vs $ |R|/|S| $]{
    		\label{fig:phase3-rsratio-t6}
        \includegraphics[width=0.22\linewidth]{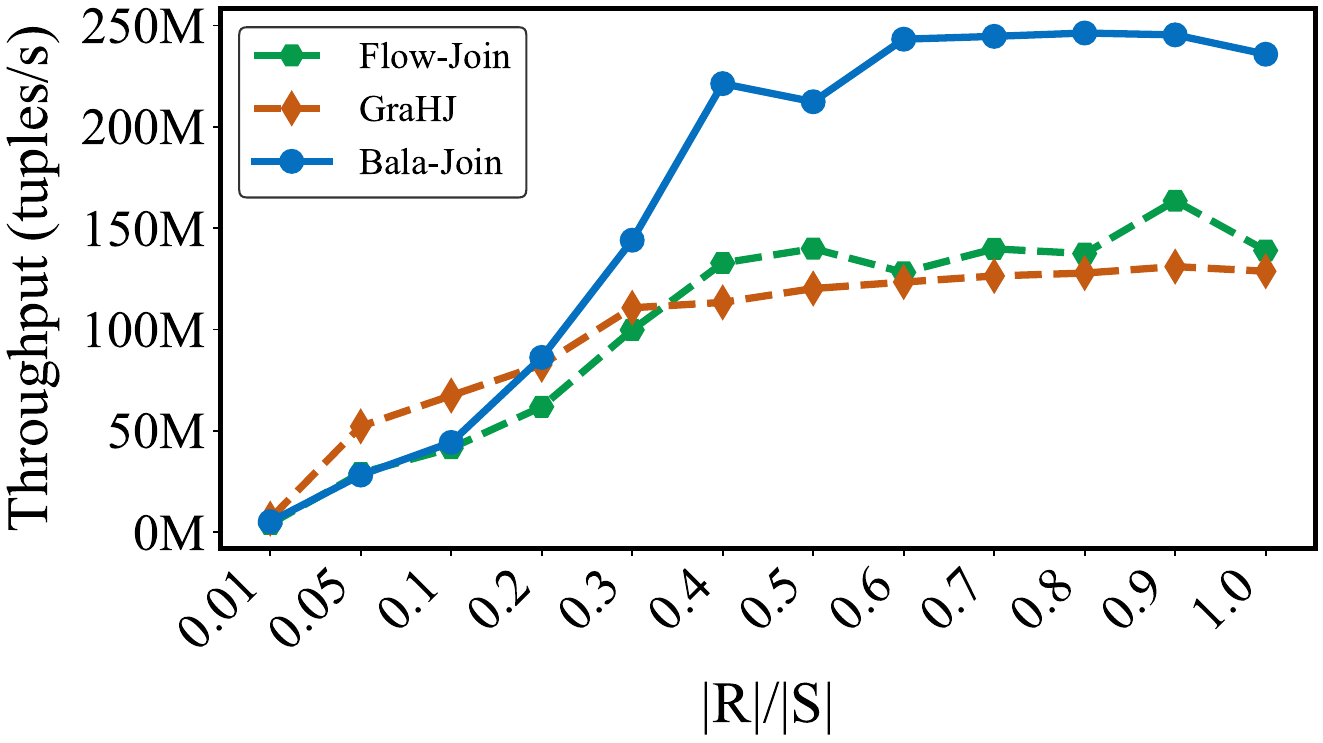}
        }\hfill
    \subfloat[Network overhead vs.  $ |R|/|S| $]{
    		\label{fig:phase3-rsratio-n6}
        \includegraphics[width=0.22\linewidth]{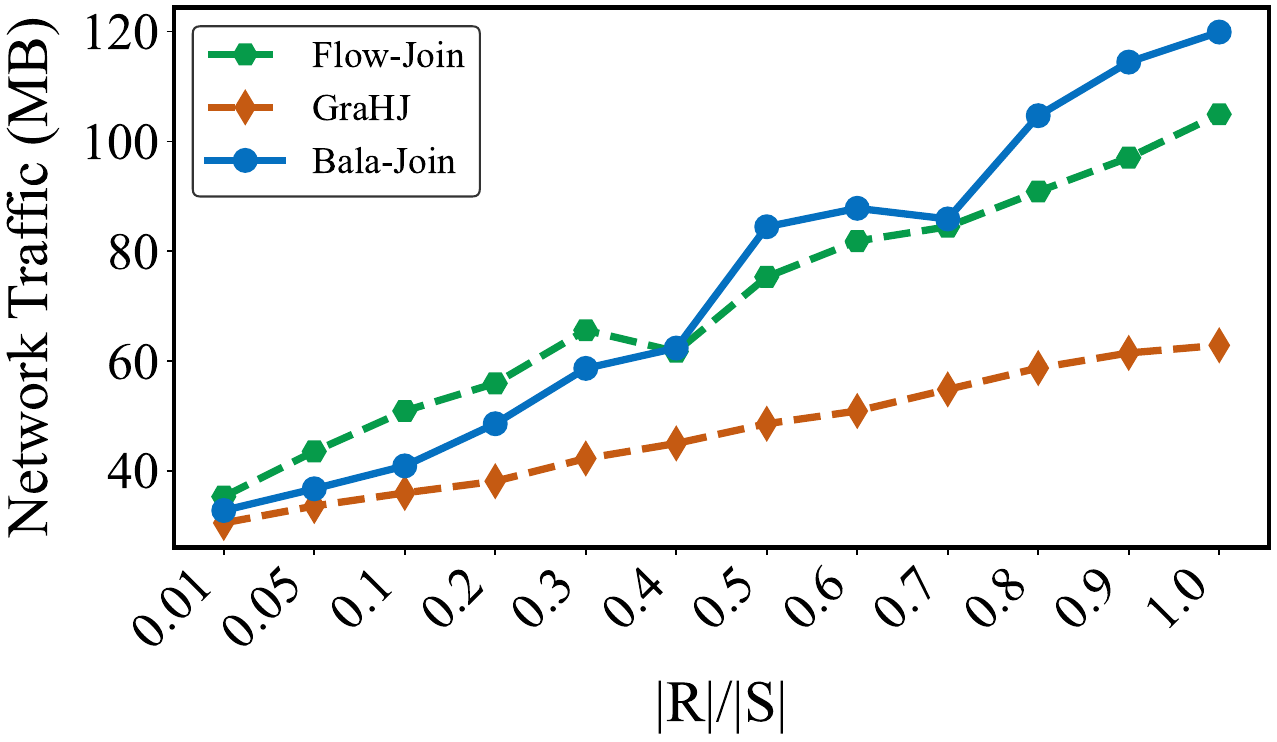}
        }  
    \caption{Comparison of \textsf{Dist-HJ} solutions in different cluster settings: (a-d) 3 nodes; (e-h) 6 nodes}
    \label{fig:phase3-6}
\end{figure*}
\begin{figure}
    \centering
    \subfloat[Throughput vs. bandwidth (3 nodes)]{
        \label{fig:phase3-bandwidth3}
        \includegraphics[width=0.74\linewidth]{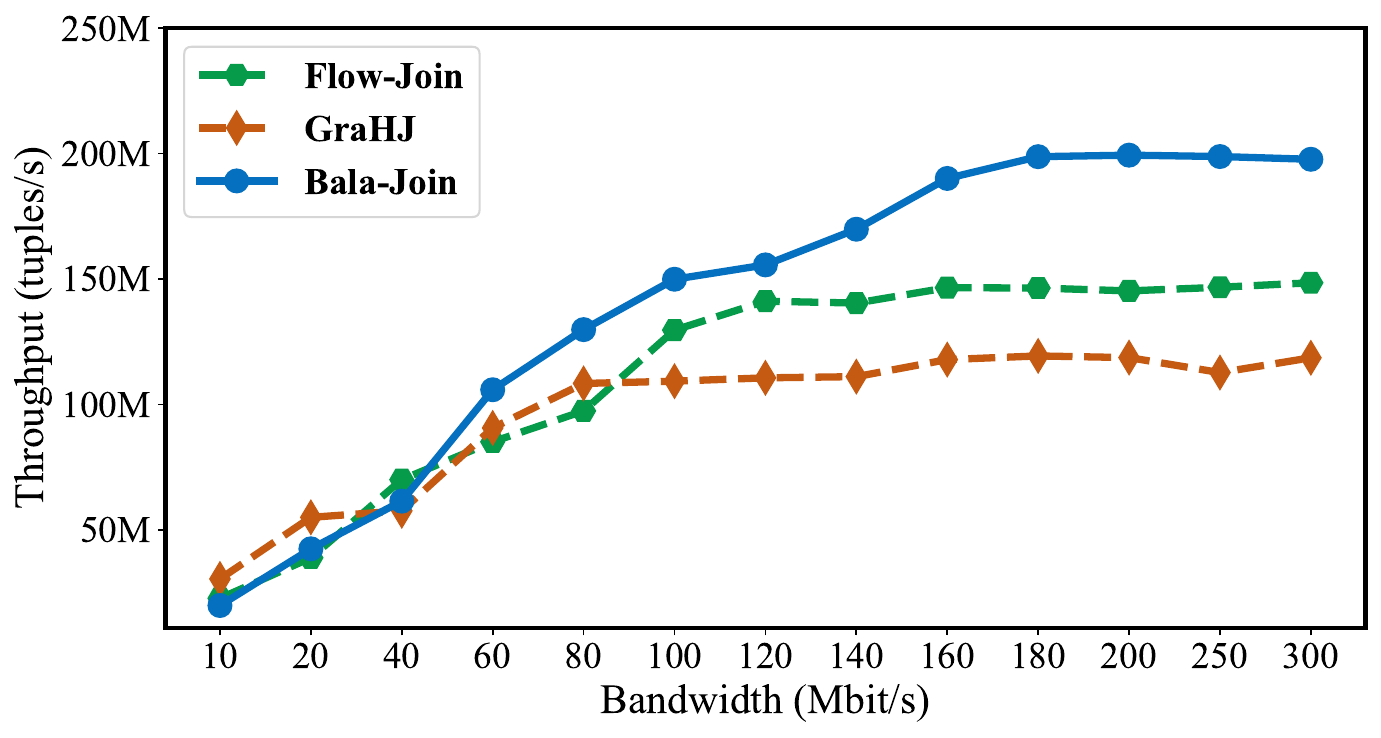}
    }\\
    
    \subfloat[Throughput vs. bandwidth (6 nodes)]{
        \label{fig:phase3-bandwidth6}
        \includegraphics[width=0.74\linewidth]{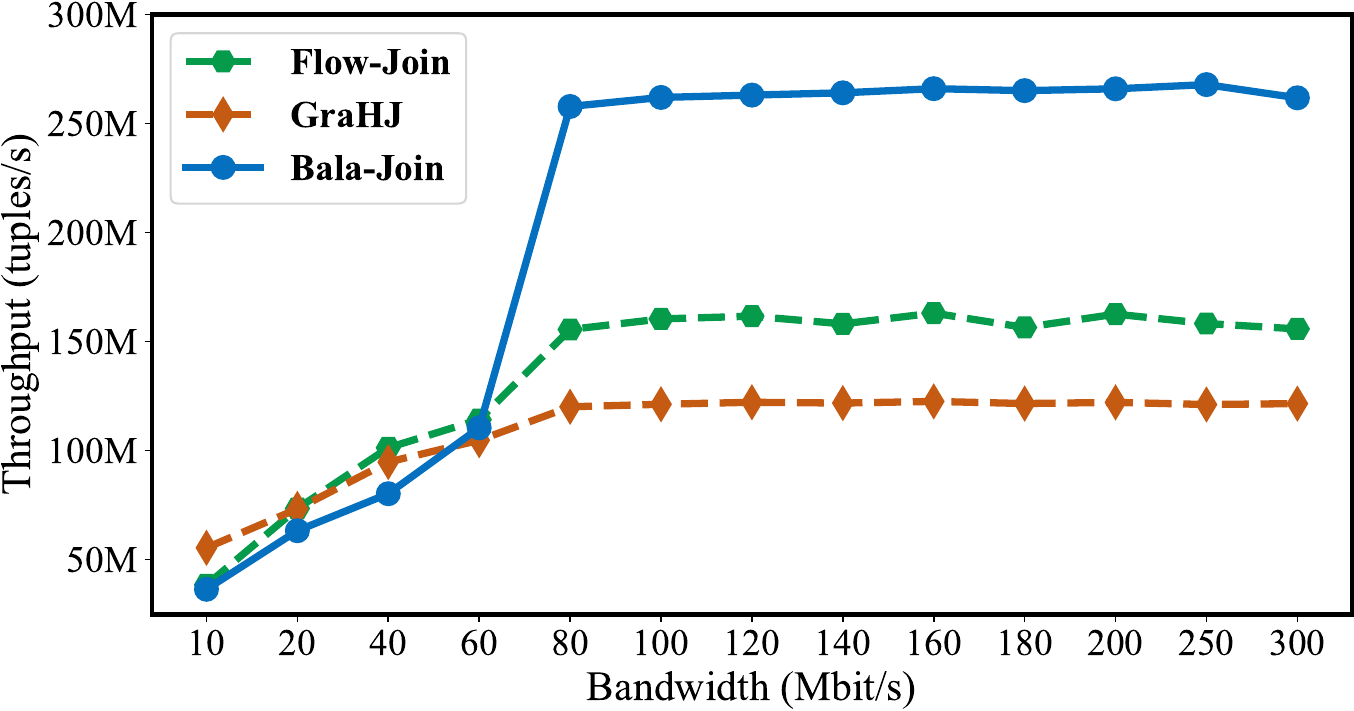}
    }
    \caption{Throughput vs. bandwidth for \textsf{Dist-HJ} solutions}
    \label{fig:phase3-3}
\end{figure}

\begin{figure}
    \centering
    
    \subfloat[LO\_ORDERDATE = D\_DATEKEY]{
        \label{fig:LD}
        \includegraphics[width=0.79\linewidth]{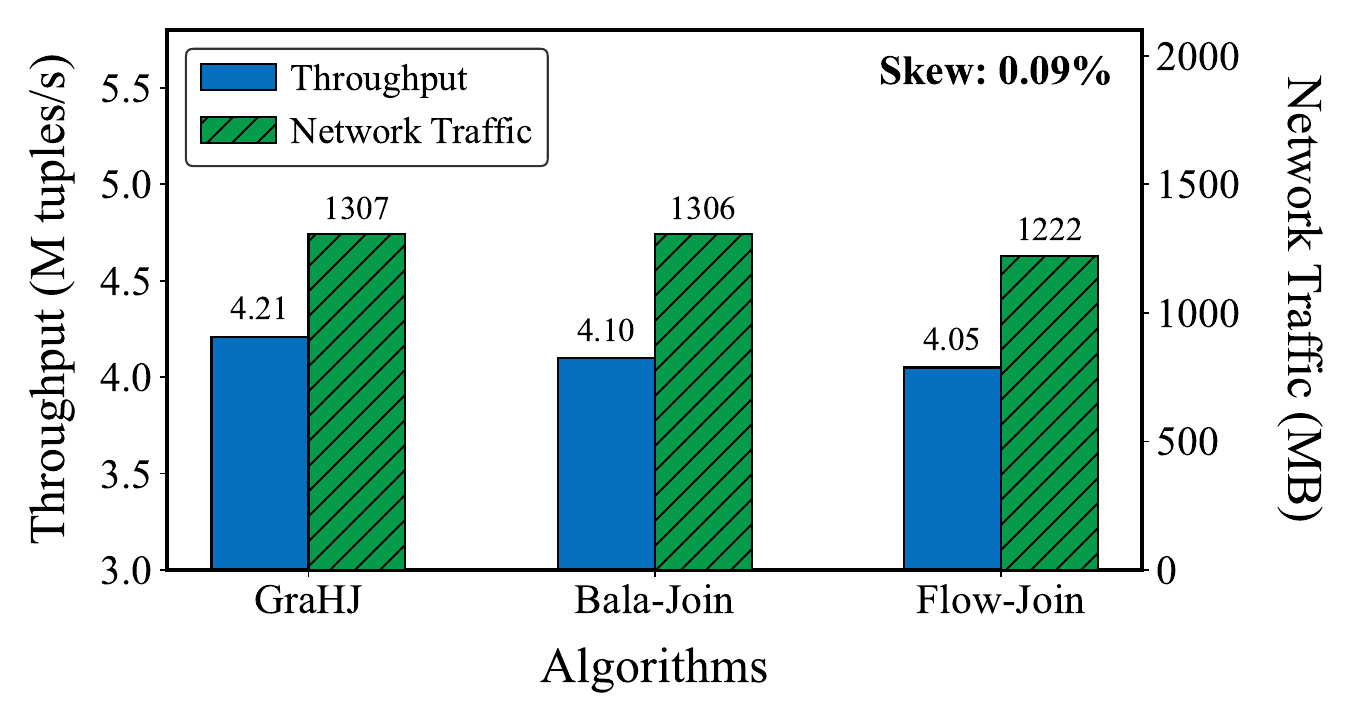}
    }\\
    \vspace{-1ex}
    
    \subfloat[LO\_SUPPKEY = S\_SUPPKEY]{
        \label{LS}
        \includegraphics[width=0.79\linewidth]{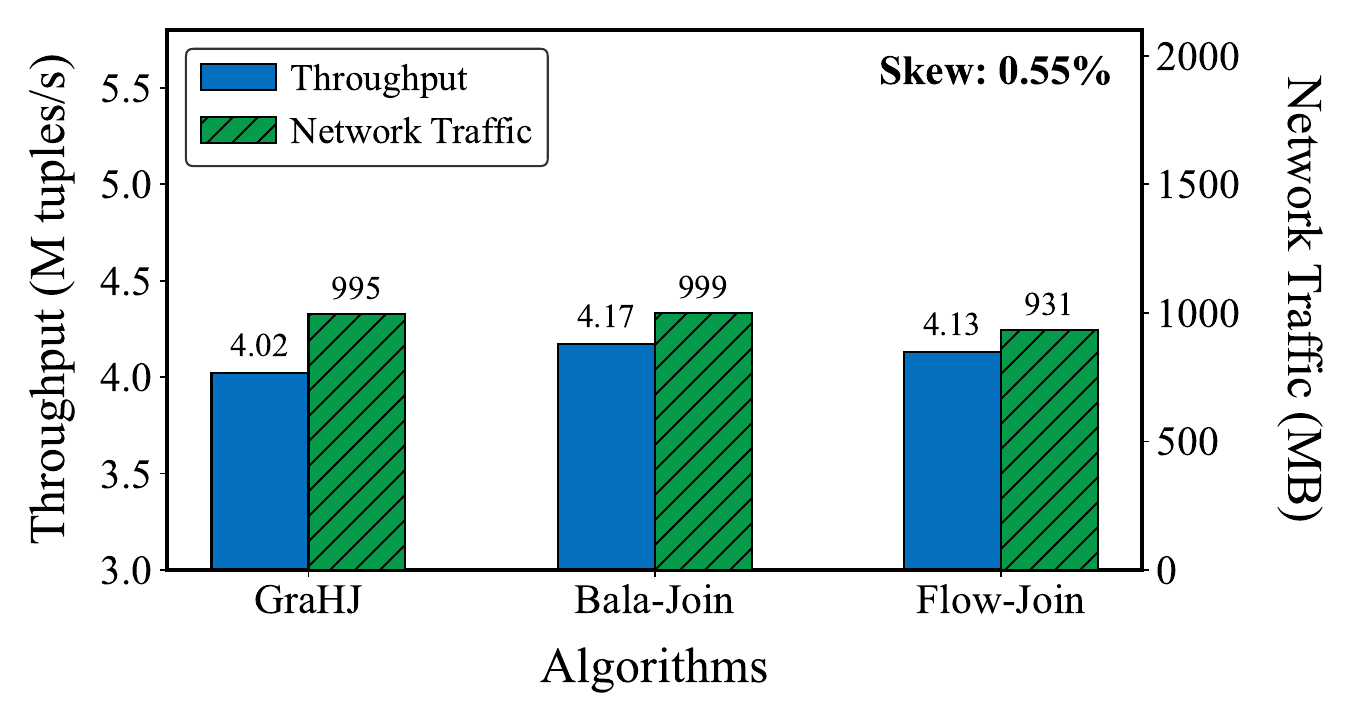}
    }
    \caption{Performance of \textsf{Dist-HJ} solutions on SSB-skew}
    \label{fig:ssb}
\end{figure}

Finally, we evaluate the complete \textsf{Dist-HJ} solution, including the entire pipeline of data skew detection, redistribution, and computation. \cref{fig:phase3-6} and \cref{fig:phase3-3} compare the performance of different solutions in 3-node and 6-node clusters, respectively. \cref{table:exp3} summarizes the average throughput per second for each solution in various scenarios, together with the performance improvements of our solution compared to \textsf{Flow-Join} and \textsf{GraHJ}.

\cref{fig:ssb} presents the results for the top two join operations, which appear most frequently in the SSB-skew benchmark, evaluated on a 24-node cluster (similar trends were observed on other cluster scales).

Both join operations involve tables with varying degrees of data skew. For the join ``LO\_ORDERDATE = D\_DATEKEY", the most frequent join key occurs at a rate of 0.09\%, while for ``LO\_SUPPKEY = S\_SUPPKEY", it occurs at 0.55\%. As the skew level increases, Bala-Join demonstrates superior performance compared to GraHJ and Flow-Join. Notably, GraHJ outperforms Bala-Join in case (a). This is because, in the SSB-skew benchmark, the same join key (LO\_ORDERDATE) is clustered together. As a result, the detector may prematurely classify early-arriving tuples as skewed, leading to unnecessary partitioning. This insight is valuable for the practical deployment of data skew detection mechanisms.

In all experimental scenarios, our comprehensive solution consistently outperforms \textsf{Flow-Join} in terms of throughput. However, it does not have a clear advantage over \textsf{Flow-Join} in terms of network overhead. Furthermore, \textsf{Flow-Join}, with its SFR strategy for skewness processing, shows a trend of reducing network overhead as the number of nodes increases from 3 to 6. However, our results demonstrate that our solution strikes a better balance between the network and the computational performance. Despite the higher network overhead, our algorithm achieves significantly better execution efficiency.

\begin{table}[h!]
\centering
\caption{Average throughput and comparison of different \textsf{Dist-HJ} implementations over all experiments in \cref{ssec54}}
\label{table:exp3}
\resizebox{\linewidth}{!}{
\begin{tabular}{c|c|c|c}
\hline
 & \textsf{Bala-Join} & \textsf{Flow-Join} & \textsf{GraHJ} \\\hline
Bandwidth (3 nodes) & $ \mathbf{140M} $ & $ 112M(+25\%) $ & $ 97M(+44\%) $ \\
Zipf factor (3 nodes) & $ \mathbf{235M} $ & $ 177M(+33\%) $ & $ 121M(+94\%) $ \\
$ |R|/|S| $ (3 nodes)  & $ \mathbf{107M} $ & $ 83M(+29\%) $ & $ 91M(+18\%) $\\
Bandwidth (6 nodes) & $ \mathbf{205M} $ & $ 135M(+52\%) $ & $ 109M(+88\%) $ \\
Zipf factor (6 nodes) & $ \mathbf{328M} $ & $ 208M(+58\%) $ & $ 121M(+171\%) $ \\
$ |R|/|S| $ (6 nodes) & $ \mathbf{163M} $ & $ 101M(+61\%) $ & $ 99M(+65\%) $\\
\hline
\end{tabular}
}
\end{table}

\section{CONCLUSION}
\label{sec:conc}
The efficiency of Distributed Hash Join (\textsf{Dist-HJ}) in geo-distributed databases, such as CockroachDB, is severely impeded by data skew, particularly over Wide Area Networks (WANs), posing a critical industrial challenge. 
In this paper, we propose Bala-Join, an adaptive solution designed to mitigate this specific performance bottleneck. 
The core of our solution, the BPPR algorithm, employs a controlled multicast mechanism to judiciously balance computational load distribution against network overhead. 
Subject to a rigorous balance factor, BPPR performs fine-grained partitioning on skewed tuples, thereby achieving a bounded load balance while minimizing the multicast scope of the build table.
This algorithm is tightly integrated with an online skew detector, which utilizes the Active-Signaling and Asynchronous-Pulling mechanism to facilitate efficient, real-time processing of data streams. 
Our empirical evaluation demonstrates that Bala-Join achieves a superior balance between network and computational performance, yielding throughput gains of 25\%-61\% and confirming its practical viability for real-world deployments.

\bibliographystyle{IEEEtran}
\bibliography{IEEEabrv,ref}

\begin{thebibliography}{10}
\providecommand{\url}[1]{#1}
\csname url@samestyle\endcsname
\providecommand{\newblock}{\relax}
\providecommand{\bibinfo}[2]{#2}
\providecommand{\BIBentrySTDinterwordspacing}{\spaceskip=0pt\relax}
\providecommand{\BIBentryALTinterwordstretchfactor}{4}
\providecommand{\BIBentryALTinterwordspacing}{\spaceskip=\fontdimen2\font plus
\BIBentryALTinterwordstretchfactor\fontdimen3\font minus \fontdimen4\font\relax}
\providecommand{\BIBforeignlanguage}[2]{{%
\expandafter\ifx\csname l@#1\endcsname\relax
\typeout{** WARNING: IEEEtran.bst: No hyphenation pattern has been}%
\typeout{** loaded for the language `#1'. Using the pattern for}%
\typeout{** the default language instead.}%
\else
\language=\csname l@#1\endcsname
\fi
#2}}
\providecommand{\BIBdecl}{\relax}
\BIBdecl

\bibitem{stonebraker1986case}
M.~Stonebraker, ``The case for shared nothing,'' \emph{IEEE Database Eng. Bull.}, vol.~9, no.~1, pp. 4--9, 1986.

\bibitem{taft2020cockroachdb}
R.~Taft, I.~Sharif, A.~Matei, N.~VanBenschoten, J.~Lewis, T.~Grieger, K.~Niemi, A.~Woods, A.~Birzin, R.~Poss \emph{et~al.}, ``Cockroachdb: The resilient geo-distributed sql database,'' in \emph{Proceedings of the 2020 ACM SIGMOD International Conference on Management of Data}, 2020, pp. 1493--1509.

\bibitem{huang2020tidb}
D.~Huang, Q.~Liu, Q.~Cui, Z.~Fang, X.~Ma, F.~Xu, L.~Shen, L.~Tang, Y.~Zhou, M.~Huang \emph{et~al.}, ``Tidb: a raft-based htap database,'' \emph{Proceedings of the VLDB Endowment}, vol.~13, no.~12, pp. 3072--3084, 2020.

\bibitem{corbett2013spanner}
J.~C. Corbett, J.~Dean, M.~Epstein, A.~Fikes, C.~Frost, J.~J. Furman, S.~Ghemawat, A.~Gubarev, C.~Heiser, P.~Hochschild \emph{et~al.}, ``Spanner: Google’s globally distributed database,'' \emph{ACM Transactions on Computer Systems (TOCS)}, vol.~31, no.~3, pp. 1--22, 2013.

\bibitem{yang2022oceanbase}
Z.~Yang, C.~Yang, F.~Han, M.~Zhuang, B.~Yang, Z.~Yang, X.~Cheng, Y.~Zhao, W.~Shi, H.~Xi \emph{et~al.}, ``Oceanbase: a 707 million tpmc distributed relational database system,'' \emph{Proceedings of the VLDB Endowment}, vol.~15, no.~12, pp. 3385--3397, 2022.

\bibitem{dewitt1992practical}
D.~J. DeWitt, J.~F. Naughton, D.~A. Schneider, and S.~Seshadri, ``Practical skew handling in parallel joins,'' University of Wisconsin-Madison Department of Computer Sciences, Tech. Rep., 1992.

\bibitem{metwally2005efficient}
A.~Metwally, D.~Agrawal, and A.~El~Abbadi, ``Efficient computation of frequent and top-k elements in data streams,'' in \emph{Database Theory-ICDT 2005: 10th International Conference, Edinburgh, UK, January 5-7, 2005. Proceedings 10}.\hskip 1em plus 0.5em minus 0.4em\relax Springer, 2005, pp. 398--412.

\bibitem{xu2008handling}
Y.~Xu, P.~Kostamaa, X.~Zhou, and L.~Chen, ``Handling data skew in parallel joins in shared-nothing systems,'' in \emph{Proceedings of the 2008 ACM SIGMOD international conference on Management of data}, 2008, pp. 1043--1052.

\bibitem{barthels2015rack}
C.~Barthels, S.~Loesing, G.~Alonso, and D.~Kossmann, ``Rack-scale in-memory join processing using rdma,'' in \emph{Proceedings of the 2015 ACM SIGMOD International Conference on Management of Data}, 2015, pp. 1463--1475.

\bibitem{rodiger2016flow}
W.~R{\"o}diger, S.~Idicula, A.~Kemper, and T.~Neumann, ``Flow-join: Adaptive skew handling for distributed joins over high-speed networks,'' in \emph{2016 IEEE 32nd International Conference on Data Engineering (ICDE)}.\hskip 1em plus 0.5em minus 0.4em\relax IEEE, 2016, pp. 1194--1205.

\bibitem{polychroniou2014track}
O.~Polychroniou, R.~Sen, and K.~A. Ross, ``Track join: distributed joins with minimal network traffic,'' in \emph{Proceedings of the 2014 ACM SIGMOD international conference on Management of data}, 2014, pp. 1483--1494.

\bibitem{yang2023one}
J.~Yang, H.~Li, Y.~Si, H.~Zhang, K.~Zhao, K.~Wei, W.~Song, Y.~Liu, and J.~Cui, ``One size cannot fit all: a self-adaptive dispatcher for skewed hash join in shared-nothing rdbmss,'' in \emph{Proceedings Of The 29th International Conference on Database Systems For Advanced Applications}, 2024.

\bibitem{pu2015low}
Q.~Pu, G.~Ananthanarayanan, P.~Bodik, S.~Kandula, A.~Akella, P.~Bahl, and I.~Stoica, ``Low latency geo-distributed data analytics,'' \emph{ACM SIGCOMM Computer Communication Review}, vol.~45, no.~4, pp. 421--434, 2015.

\bibitem{pradhan2024optimal}
A.~Pradhan, S.~Karthik, and R.~S, ``Optimal query plans for geo-distributed data analytics at scale,'' in \emph{Proceedings of the 7th Joint International Conference on Data Science \& Management of Data (11th ACM IKDD CODS and 29th COMAD)}, 2024, pp. 247--251.

\bibitem{kitsuregawa1990bucket}
M.~Kitsuregawa and Y.~Ogawa, ``Bucket spreading parallel hash: A new, robust, parallel hash join method for data skew in the super database computer (sdc).'' in \emph{VLDB}, 1990, pp. 210--221.

\bibitem{vitorovic2016load}
A.~Vitorovic, M.~Elseidy, and C.~Koch, ``Load balancing and skew resilience for parallel joins,'' in \emph{2016 IEEE 32nd International Conference on Data Engineering (ICDE)}.\hskip 1em plus 0.5em minus 0.4em\relax Ieee, 2016, pp. 313--324.

\bibitem{gibbons1998new}
P.~B. Gibbons and Y.~Matias, ``New sampling-based summary statistics for improving approximate query answers,'' in \emph{Proceedings of the 1998 ACM SIGMOD international conference on Management of data}, 1998, pp. 331--342.

\bibitem{manku2002approximate}
G.~S. Manku and R.~Motwani, ``Approximate frequency counts over data streams,'' in \emph{VLDB'02: Proceedings of the 28th International Conference on Very Large Databases}.\hskip 1em plus 0.5em minus 0.4em\relax Elsevier, 2002, pp. 346--357.

\bibitem{demaine2002frequency}
E.~D. Demaine, A.~L{\'o}pez-Ortiz, and J.~I. Munro, ``Frequency estimation of internet packet streams with limited space,'' in \emph{Esa}, vol.~2.\hskip 1em plus 0.5em minus 0.4em\relax Citeseer, 2002, pp. 348--360.

\bibitem{charikar2004finding}
M.~Charikar, K.~Chen, and M.~Farach-Colton, ``Finding frequent items in data streams,'' \emph{Theoretical Computer Science}, vol. 312, no.~1, pp. 3--15, 2004.

\bibitem{cormode2005improved}
G.~Cormode and S.~Muthukrishnan, ``An improved data stream summary: the count-min sketch and its applications,'' \emph{Journal of Algorithms}, vol.~55, no.~1, pp. 58--75, 2005.

\bibitem{cafaro2016parallel}
M.~Cafaro, M.~Pulimeno, and P.~Tempesta, ``A parallel space saving algorithm for frequent items and the hurwitz zeta distribution,'' \emph{Information Sciences}, vol. 329, pp. 1--19, 2016.

\bibitem{zhao2023double}
Y.~Zhao, W.~Han, Z.~Zhong, Y.~Zhang, T.~Yang, and B.~Cui, ``Double-anonymous sketch: Achieving top-k-fairness for finding global top-k frequent items,'' \emph{Proceedings of the ACM on Management of Data}, vol.~1, no.~1, pp. 1--26, 2023.

\bibitem{walton1991taxonomy}
C.~B. Walton, A.~G. Dale, and R.~M. Jenevein, ``A taxonomy and performance model of data skew effects in parallel joins,'' in \emph{Proceedings of the 17th International Conference on Very Large Data Bases}, 1991, pp. 537--548.

\bibitem{hua1991handling}
K.~A. Hua and C.~Lee, ``Handling data skew in multiprocessor database computers using partition tuning.'' in \emph{VLDB}, vol.~91, 1991, pp. 525--535.

\bibitem{alsabti2001skew}
K.~Alsabti and S.~Ranka, ``Skew-insensitive parallel algorithms for relational join,'' \emph{Journal of King Saud University-Computer and Information Sciences}, vol.~13, pp. 79--110, 2001.

\bibitem{wolf1991effective}
J.~L. Wolf, D.~M. Dias, and P.~S. Yu, ``An effective algorithm for parallelizing hash joins in the presence of data skew,'' in \emph{Proceedings. Seventh International Conference on Data Engineering}.\hskip 1em plus 0.5em minus 0.4em\relax IEEE Computer Society, 1991, pp. 200--201.

\bibitem{wolf1993parallel}
------, ``A parallel sort merge join algorithm for managing data skew,'' \emph{IEEE Transactions on Parallel and Distributed Systems}, vol.~4, no.~1, pp. 70--86, 1993.

\bibitem{wolf1994new}
J.~L. Wolf, D.~M. Dias, P.~S. Yu, and J.~Turek, ``New algorithms for parallelizing relational database joins in the presence of data skew,'' \emph{IEEE Transactions on Knowledge and Data Engineering}, vol.~6, no.~6, pp. 990--997, 1994.

\bibitem{dewan1994predictive}
H.~M. Dewan, K.~W. Mok, M.~Hern{\'a}ndez, and S.~J. Stolfo, ``Predictive dynamic load balancing of parallel hash-joins over heterogeneous processors in the presence of data skew,'' in \emph{Proceedings of 3rd International Conference on Parallel and Distributed Information Systems}.\hskip 1em plus 0.5em minus 0.4em\relax IEEE, 1994, pp. 40--49.

\bibitem{graham1969bounds}
R.~L. Graham, ``Bounds on multiprocessing timing anomalies,'' \emph{SIAM journal on Applied Mathematics}, vol.~17, no.~2, pp. 416--429, 1969.

\bibitem{coffman1978application}
E.~G. Coffman, Jr, M.~R. Garey, and D.~S. Johnson, ``An application of bin-packing to multiprocessor scheduling,'' \emph{SIAM Journal on Computing}, vol.~7, no.~1, pp. 1--17, 1978.

\bibitem{shatdal1993using}
A.~Shatdal and J.~F. Naughton, ``Using shared virtual memory for parallel join processing,'' in \emph{Proceedings of the 1993 ACM SIGMOD international conference on Management of data}, 1993, pp. 119--128.

\bibitem{kitsuregawa1995dynamic}
M.~Kitsuregawa, ``Dynamic join product skew handling for hash-joins in shared-nothing database systems,'' in \emph{Proceedings Of The Fourth International Conference on Database Systems For Advanced Applications}, vol.~5.\hskip 1em plus 0.5em minus 0.4em\relax World Scientific, 1995, p. 246.

\bibitem{zhou1995handling}
X.~Zhou and M.~E. Orlowska, ``Handling data skew in parallel hash join computation using two-phase scheduling,'' in \emph{Proceedings 1st International Conference on Algorithms and Architectures for Parallel Processing}, vol.~2.\hskip 1em plus 0.5em minus 0.4em\relax IEEE, 1995, pp. 527--536.

\bibitem{zhou2019fastjoin}
S.~Zhou, F.~Zhang, H.~Chen, H.~Jin, and B.~B. Zhou, ``Fastjoin: A skewness-aware distributed stream join system,'' in \emph{2019 IEEE International Parallel and Distributed Processing Symposium (IPDPS)}.\hskip 1em plus 0.5em minus 0.4em\relax IEEE, 2019, pp. 1042--1052.

\bibitem{gavagsaz2019load}
E.~Gavagsaz, A.~Rezaee, and H.~Haj Seyyed~Javadi, ``Load balancing in join algorithms for skewed data in mapreduce systems,'' \emph{The Journal of Supercomputing}, vol.~75, pp. 228--254, 2019.

\bibitem{kitsuregawa1983application}
M.~Kitsuregawa, H.~Tanaka, and T.~Moto-Oka, ``Application of hash to data base machine and its architecture,'' \emph{New Generation Computing}, vol.~1, no.~1, pp. 63--74, 1983.

\bibitem{huang2014ld}
Q.~Huang and P.~P. Lee, ``Ld-sketch: A distributed sketching design for accurate and scalable anomaly detection in network data streams,'' in \emph{IEEE INFOCOM 2014-IEEE Conference on Computer Communications}.\hskip 1em plus 0.5em minus 0.4em\relax IEEE, 2014, pp. 1420--1428.

\bibitem{stamos1993symmetric}
J.~W. Stamos and H.~C. Young, ``A symmetric fragment and replicate algorithm for distributed joins,'' \emph{IEEE Transactions on Parallel and Distributed Systems}, vol.~4, no.~12, pp. 1345--1354, 1993.

\bibitem{zipf1949human}
G.~K. Zipf, \emph{Human behavior and the principle of least effort: An introduction to human ecology}.\hskip 1em plus 0.5em minus 0.4em\relax Addison-Wesley, 1949.

\bibitem{rodiger2014locality}
W.~R{\"o}diger, T.~M{\"u}hlbauer, P.~Unterbrunner, A.~Reiser, A.~Kemper, and T.~Neumann, ``Locality-sensitive operators for parallel main-memory database clusters,'' in \emph{2014 IEEE 30th International Conference on Data Engineering}.\hskip 1em plus 0.5em minus 0.4em\relax IEEE, 2014, pp. 592--603.

\bibitem{justen2024polar}
D.~Justen, D.~Ritter, C.~Fraser, A.~Lamb, A.~Lee, T.~Bodner, M.~Y. Haddad, S.~Zeuch, V.~Markl, and M.~Boehm, ``Polar: Adaptive and non-invasive join order selection via plans of least resistance,'' \emph{Proceedings of the VLDB Endowment}, vol.~17, no.~6, pp. 1350--1363, 2024.

\bibitem{o2007star}
P.~E. O’Neil, E.~J. O’Neil, and X.~Chen, ``The star schema benchmark (ssb),'' \emph{Pat}, vol. 200, no.~0, p.~50, 2007.

\end{thebibliography}

\end{document}